\documentclass[journal]{IEEEtran}
\usepackage{amsmath,amsfonts}
\usepackage{algorithmic}
\usepackage{algorithm}
\usepackage{array}
\usepackage{subfig}
\usepackage{textcomp}
\usepackage{stfloats}
\usepackage{url}
\usepackage{verbatim}
\usepackage{graphicx}
\usepackage{cite}
\usepackage{amsthm}
\usepackage{amsmath}
\usepackage{booktabs}

\usepackage{caption}
\DeclareCaptionLabelFormat{tablelabel}{\footnotesize{TABLE \thetable}} 
\newtheorem{proposition}{Proposition}

\theoremstyle{remark}
\newtheorem{remark}{Remark}

\usepackage{enumitem}
\begin{document}

\title{Flexible Opteration of Electricity-HCNG Networks with Adjustable Hydrogen Fraction via Distributionally Robust Method}
\title{Distributionally Robust Chance-Constrained Optimal Dispatch of Electricty-HCNG Systems with Variable Hydrogen Fraction}

\title{Flexible Operation of Electricity-HCNG Networks with
	Variable Hydrogen Fraction: A Distributionally Robust Chance-Constrained Approach}

\title{Flexible Operation of Electricity-HCNG Networks with
	Variable Hydrogen Fraction: A Joint Distributionally Robust Chance-Constrained Method}
	
\title{Flexible Operation of Electricity-HCNG Networks with
	Variable Hydrogen Fraction: A Distributionally Robust Joint Chance-Constrained Approach}
	

	
	

	

        
\author{
	\IEEEauthorblockN{Sicheng Liu,~\IEEEmembership{Student Member,~IEEE,}}
	\IEEEauthorblockN{Bo Yang,~\IEEEmembership{Senior Member,~IEEE,}}
	\IEEEauthorblockN{Xu Yang,~\IEEEmembership{Student Member,~IEEE,}}
	\IEEEauthorblockN{Xin Li}
	\IEEEauthorblockN{Zhaojian Wang,~\IEEEmembership{Member,~IEEE,}}
	\IEEEauthorblockN{Xinping Guan,~\IEEEmembership{Fellow,~IEEE}}
}

\markboth{Journal of \LaTeX\ Class Files,~Vol.~14, No.~8, August~2021}%
{Shell \MakeLowercase{\textit{et al.}}: A Sample Article Using IEEEtran.cls for IEEE Journals}


\maketitle

\begin{abstract}
Hydrogen-enriched compressed natural gas (HCNG) is a promising way to utilize surplus renewable energy through hydrogen electrolysis and blending it into natural gas.
However, the optimal hydrogen volume fraction (HVF) of HCNG varies following the daily fluctuations of renewable energy.
Besides, facing the rapid volatility of renewable energy, 
ensuring rapid and reliable real-time adjustments is challenging for electricity-HCNG (E-HCNG) coupling networks.
To this end, this paper proposes a flexible operation framework for electricity-HCNG (E-HCNG) networks against the fluctuations and volatility of renewable energy.
Based on operations with variable HVF, the framework developed an E-HCNG system-level affine policy, which allows real-time re-dispatch of operations according to the volatility.
Meanwhile, to guarantee the operational reliability of the affine policy,
a distributionally robust joint chance constraint (DRJCC) is introduced, which limits the violation probability of operational constraints under the uncertainties of renewable energy volatility.
Furthermore, in the solving process, to mitigate the over-conservation in DRJCC decomposition, an improved risk allocation method is proposed, utilizing the correlations among violations under the affine policy. 
Moreover, to tackle the non-convexities arising from the variable HVF, customized approximations for HCNG flow formulations are developed.
The problem is finally reformulated into a mix-integer second-order cone programming problem.
The effectiveness of the proposed method is validated both in small-scale and large-scale experiments.
\end{abstract}


\vspace{-5pt}
\addcontentsline{toc}{section}{Nomenclature}
\section*{Nomenclature}
\vspace{-5pt}
\small
\subsection{Indices and Sets}
\begin{IEEEdescription}[\IEEEusemathlabelsep\IEEEsetlabelwidth{$HHV^{Hy},HHV^{Ga}$}] 
	\item[$i,j,h$] Index of grid bus.
	\item[$m,n,o$] Index of HCNG network node.
	\item[$ij$] Index of power line.
	
	\item[$mn$] Index of hydrogen pipeline.
	\item[$t$] Index of time period.
	\item[$\Omega _{(\cdot)}$] Set of energy components $(\cdot)$.
	\item[$\Phi^0$] Set of deterministic scheduling decisions.
	\item[$\Phi$] Set of all scheduling decisions containing affine factors.
	\item[$\mathcal{D}$] The ambiguity set of PV output prediction errors.
\end{IEEEdescription}
\vspace{-15pt}
\subsection{Parameters}
\begin{IEEEdescription}[\IEEEusemathlabelsep\IEEEsetlabelwidth{$HHV^{Hy},HHV^{G}$}] 
	\item[$a_{i}^{\mathrm{NG}},b_{i}^{\mathrm{NG}},c_{i}^{\mathrm{NG}}$] Generation cost factors of non-GFU $i$. 
	\item[$c_{m}^{S}$] Cost factor of gas source $m$.
	\item[$\beta_{m}^{\mathrm{P2H}}$] Conversion efficiency of P2H $m$.
	\item[$\beta_m^G$] Generation efficiency of GFU $m$.
	\item[$v_{k}^{C}$] Loss factor of HCNG compressor $k$.
	\item[$K_{mn}$] Line pack coefficients of HCNG pipeline $mn$.
	\item[$M^{Hy},M^{Gas}$] Molar value of hydrogen and natural gas.
	\item[$Heat_{m,t}^{\mathrm{L}}$] The heat load of HCNG node $m$ at $t$.
	\item[$HHV^{Hy},HHV^{Gas}$] The heat value of hydrogen and natural gas.
	\item[$P_{i,t}^{\mathrm{PV}}$] Predicted output of PV $i$ at $t$.
    \item[$\boldsymbol{\xi}_t$] The PV output predicted error at $t$.
	\item[$\boldsymbol{e}$] Vector of all ones.
	\item[$\eta _{m}^{ch},\eta _{m}^{dch}$] Charging and discharging efficiency of hydrogen storage $m$.
	\item[$c^{(\cdot),U},c^{(\cdot),D}$] Cost of upward and downward reserve capacity by corresponding component $(\cdot)$.
	\item[$(\cdot)^{\mathrm{max}},(\cdot)^{\mathrm{min}}$] Maximum and minimum bounds for corresponding operation $(\cdot)$ .
	\item[$\epsilon$] Operational violation risk of DRJCC.
\end{IEEEdescription}
\vspace{-10pt}
\subsection{Variables}
\begin{IEEEdescription}[\IEEEusemathlabelsep\IEEEsetlabelwidth{$HHV^{Hy},HHV^{G}$}] 
	\item[$F^0$] Cost function of deterministic scheduling.
	\item[$P_{i,t}^{\mathrm{NG}},P_{i,t}^{\mathrm{G}}$] Power outputs of non-GFU $i$ and GFU $i$ at $t$.
	\item[$P_{i,t}^{\mathrm{P2H}}$] Power consumption of P2H $i$ at $t$.
	\item[$P_{i,t}^{\mathrm{L}}$] Power load of bus $i$ at $t$.
	\item[$\mathrm{GSF}_{l,i}$] Generation swift factor from bus $i$ to line $l$.
	\item[$F_{m,t}^{\mathrm{S}}$] Supply volume of natural gas source $m$ at $t$.
	\item[$F_{m,t}^{\mathrm{P2H}}$] Hydrogen production by P2H $m$ at $t$.
	\item[$F_{m,t}^{\mathrm{dch}}$] Hydrogen released from storage $m$ at $t$.
	\item[$E_{m,t}$] Hydrogen amount of hydrogen storage $m$ at $t$.
	\item[$F_{m,t}^{\mathrm{G}}$] HCNG consumption of GFU $m$ at $t$.
	\item[$F_{m,t}^{\mathrm{L}}$] HCNG demand of node $m$ at $t$, influenced by the HVF.
	\item[$F_{k,t}^{C}$] HCNG flow through compressor $k$ at $t$.
	\item[$F_{mm,t}$] HCNG flow of pipeline $mn$ at $t$.
	\item[$F_{mm,t}^{\mathrm{in}},F_{mm,t}^{\mathrm{out}}$] In- and out-HNCG flow of pipeline $mn$ at $t$.
	\item[$C_{mn,t}$] Coefficient for flow-pressure relationship of pipeline $mn$ at $t$, influenced by the HVF.
	\item[$\pi_{m,t}$] HCNG pressure of node $m$ at $t$.
	\item[$\pi_{m,t}^{\mathrm{in}},\pi_{m,t}^{\mathrm{out}}$] In- and out-pressure of compressor $k$ at $t$.s
	\item[$LP_{mn,t}$] Line pack of pipeline $mn$ at $t$.
	\item[$M_t$] Average molar value of HCNG at $t$, influenced by the HVF.
	\item[$HHV_t$] Heat value of HCNG at $t$, influenced by the HVF.
	\item[$w_t$] Hydrogen volume fraction (HVF) of HCNG at $t$.
	\item[$\alpha^{(\cdot)}$] Adjustment factor of corresponding scheduling $(\cdot)$.
	\item[$\tilde{(\cdot)}$] Actual execution volume of corresponding scheduling $(\cdot)$ after adjustment.
	\item[$R_{i,t}^{(\cdot),\mathrm{U}},R_{i,t}^{(\cdot),\mathrm{D}}$]  Upward and downward reserve capacity of corresponding component $(\cdot)$ $i$ at $t$.

\end{IEEEdescription}
\section{Introduction}


\IEEEPARstart{A}{s} the pace of global energy transition accelerates, hydrogen is gaining increasing attention. 
Hydrogen can be produced through electrolysis, thereby enhancing the utilization of renewable energy sources, such as solar and wind power.
In terms of usage, hydrogen is expected to serve as a replacement for traditional fossil fuels, aiming to decarbonize sectors such as transportation, chemicals, and power generation.

However, storage and transportation of hydrogen remain costly and pose safety risks \cite{c.2013Economic}. 
Moreover, the size of the consumer market for hydrogen is still relatively small \cite{chen2020GuanDaoXianQing1}. 
Fortunately, hydrogen-enriched compressed natural gas (HCNG) is recognized as an effective solution to the challenges of transportation, storage, and usage.
By blending hydrogen into natural gas, HCNG enables low-cost transportation, storage, and widespread usage through existing natural gas infrastructures.
So far, many demonstration projects about HCNG have been launched around the world, such as HYREADY \cite{onno2016HYREADY}, Sustainable Ameland \cite{mjHYDROGEN}, GRHYD \cite{GRHYD}, and Chaoyang renewable hydrogen blending \cite{GuoNeiShouLi}. 

 

While the economic benefits of electric-HCNG (E-HCNG) systems are widely validated \cite{c.2013Economica,christopherj.2018Powertogas},
their optimal operation is complicated and threatened by the hydrogen blending, compared to conventional multi-energy systems.
Specifically, the properties of HCNG, such as density and heat value, change with hydrogen blending, which increases the complexity of operational models like HCNG flow and heat load.  
%
%

In response to the challenge of complexities, several works investigated the optimal dispatch of E-HCNG networks.
Fu and Lin et al. were pioneers in studying the optimal operation of E-HCNG distribution systems.
They characterized how hydrogen blending affects the gas flow and heat loads \cite{chen2020GuanDaoXianQing1}.
Subsequently, Zhou et al. developed an operational model for E-HCNG transmission networks, considering the gas storage capacities of line packs. They also introduced an iterative algorithm for solutions \cite{yue2023HCNGZuiYouDiaoDuMoXing}.
These studies provide valuable references for modeling E-HCNG systems with a constant hydrogen volume fraction (HVF), which refers to the ratio of hydrogen volume to the total HCNG volume.
However, for systems with significant fluctuations in supply and demand, operational flexibility is limited by the constant HVF, meaning that hydrogen blending must remain in a fixed ratio to the natural gas supply.
This limitation of the HVF mode may result in increased operating costs for E-HCNG networks.


To achieve optimal operation under variable HVF conditions, existing studies employed various methods to estimate the dynamics of HVF with scheduling operation.
Research \cite{isam2021JieDianJiHVFJianMoHuLueline} estimated the mixed gas fraction including HVF at each node by calculating the proportion of gas inflow.
Similarly, resilience-oriented scheduling for integrated energy systems using HCNG and variable HVF was proposed, which significantly enhances resilience \cite{ran2024KeBianXianQingBiLi}.
Subsequently, another approach introduced dynamic HVF estimation using advective transport differential equations, continuously characterizing HVF variations \cite{sleiman2022YouHVFKuoSanMoXing}.
Moreover, Jiang et al. proposed a global HVF estimation method, which combined the line pack effect on HCNG concentration. 
And they also constructed an optimal E-HCNG scheduling strategy in day-ahead \cite{yunpeng2024ChongDa}. 
Despite the various methods for estimating HVF, a significant challenge remains: the rapid volatility of renewable energy output, such as solar and wind, cannot be accurately predicted.
This volatility challenges the energy balance and economic efficiency of E-HCNG network operations with variable HVF.


Facing the challenge of renewable energy volatility, even though the research on E-HCNG scheduling is still in the early stages, 
the existing studies for integrated energy systems can provide meaningful references. 
To realize flexible operations in response to the volatility of renewable energy, an affine policy for power-gas scheduling was proposed \cite{cheng2019LiangJieDuanDROaffine}. 
The affine policy enables real-time re-dispatch of operations after observing the renewable energy deviations. 
However, the variable HVF introduces high-degree nonconvexities into the formulation of HCNG flow, load, etc.
These nonconvexities increase the difficulty of constructing affine policies for E-HCNG networks.


Moreover, to ensure the reliability about the affine policy, 
distributionally robust chance-constrains (DRCCs) were introduced for the scheduling of power-gas integrated system, under the uncertainties of renewable energy forecast errors \cite{lun2020DianQiOPF}.
%
%
DRCCs can restrict the violation probability of operational constraints within the desired risk level, even under the worst-case probability distribution of renewable energy uncertainties.
Further, to achieve system-wide reliability, i.e., managing the joint violation probability of all operational constraints, 
distributionally robust joint chance constraint (DRJCC) is introduced \cite{Distributionallyrobustjoint}. 
DRJCC allows operators to set the violation risk at the system level, eliminating the need for complex individual analysis of each constraint.
%
To solve the DRJCC, the averaged and optimized Bonferroni approximation were successively proposed and widely adopted \cite{doi:10.1137/050622328, lun2022YangLun}.
These methods allocate the joint risk level across individual operational constraints in DRJCCs, thereby approximately decomposing them into independent DRCCs and making the problem more tractable.
However, the Bonferroni approximation assumes that the violations of each constraint are probabilistically mutually exclusive \cite{weijun2022LiLun}.
Under the affine policy, the violation probabilities of constraints are highly correlated. 
This gap leads to potential over-conservation and reduces operational economic efficiency, especially when there are numerous constraints in DRJCC, like E-HCNG networks.
In summary, although several works contribute to the optimal operations of E-HCNG networks,
considering the fluctuations and volatility of renewable energy,
there are still some gaps as follows:

\begin{enumerate}
	\item {

	\textit{In terms of operating formulation}: 
	While several works studies the operations with variable HVF to follow the daily fluctuations of renewable energy, the rapid volatility threatens the energy balance and economic efficiency of operations, highlighting the need for efficient real-time re-dispatch operation methods.
	}

	\item{
	\textit{In terms of problem solving}: Although the affine policy combined with DRJCC is promising to achieve quick and reliable operations, there are still some technical problems.
	First, existing DRJCC decomposition methods overlook the correlations of violations under the affine policy, which can bring over-conservatism, especially for complicated systems like E-HCNG networks.
	Second, the variable HVF introduces several higher-degree nonconvexities into the formulation, making the problem intractable.
	Therefore, a more efficient and customized solution method is necessary.	
	}
\end{enumerate}

Facing the above issues, this paper proposes a flexible optimal operation method for E-HCNG networks.
The detailed contributions are as follows:

\begin{enumerate}
	\item{
		
		A flexible operation framework for E-HCNG networks is proposed, enhancing energy balance and operational economic efficiency.		
		The framework not only sufficiently characterizes variable HVF and its impact on HCNG flow and load, but also proposes a system-level affine policy for real-time re-dispatch of operations against renewable energy volatility.

	}
	\item{
		 	To guarantee the operational reliability of the affine policy, DRJCC is introduced to limit the overall violation probability of operational constraints, based on an unimodality-skewness informed ambiguity set to characterize the renewable energy uncertainties.
			More importantly, to reduce the over-conservation of DRJCC decomposition in the solving process, an improved risk allocation method is proposed, utilizing the correlations of violations under the affine policy.

}
	\item{
	To tackle the high-order nonconvexities introduced by the variable HVF, first, the bilinear terms are addressed by the binary expansion approximation, maintaining the consistency of HVF.
	Moreover, for the nonconvex cubic HCNG flow constraints, a series of customized approximations via the concave-convex procedure with multiple auxiliary variables is proposed.
	Ultimately, the problem is transformed into a mixed-integer second-order cone programming problem, which can be efficiently solved.
	}

\end{enumerate}

The remainder of this paper is organized as follows: 
Section II  constructs the formulation of E-HCNG operation with variable HVF based on DRJCC and system-level re-dispatch affine policy.
Section III develops the solutions for DRJCC and non-convexity terms.
Case studies are presented in Section IV to demonstrate the effectiveness of the proposed method.
Finally, Section V concludes the paper.






\section{Problem Formultaion}
\vspace{0pt}
The configuration of the E-HCNG network system is illustrated in Fig. \ref{fig_config}. 
Hydrogen stations, which are key components of this system, connect the grid to the HCNG network.
These stations are equipped with power-to-hydrogen (P2H) electrolyzers and hydrogen storage facilities \cite{chen2020GuanDaoXianQing1}.
The produced hydrogen is buffered, stored, and then blended into gas pipelines according to the schedule.



In this section, the deterministic scheduling formulation for the E-HCNG network is first presented to illustrate the basic mode of operations. 
Subsequently, based on the deterministic model, a system-level affine policy is established to re-dispatch operations in response to renewable energy volatility.
Later on, to ensure the operational reliability of affine policy, the DRJCC is introduced with respect to renewable energy uncertainties. 
Finally, a unified E-HCNG operation model is formulated.
\vspace{-8pt}
\subsection{Deterministic E-HCNG Scheduling Formulation}
To demonstrate the operation mechanism of E-HCNG networks, the deterministic model is first established without incorporating the affine policy and reserve cost, as detailed below. 

\subsubsection{Objective Function}
The objective of the scheduling is to minimize the operating cost, which can be expressed as
\vspace{-3pt}
\begin{align*}
\underset{\Phi ^0}{\mathrm{Min}}\,\,F^0=&\sum_{t=1}^T{\sum_{i\in \Omega _{\mathrm{NG}}}{\left( a_{i}^{\mathrm{NG}}\left( P_{i,t}^{\mathrm{NG}} \right) ^2+b_{i}^{\mathrm{NG}}P_{i,t}^{\mathrm{NG}}+c_{i}^{\mathrm{NG}} \right)}}
\\
&+\sum_{t=1}^T{\sum_{m\in \Omega _{\mathrm{S}}}{c_{m}^{S}}}F_{m,t}^{S},
 \tag{1}
\end{align*}

\vspace{-6pt}
\noindent where $\varPhi ^0$ represents the set of deterministic operational decisions for both the grid and the HCNG network. The first and second terms of $F^0$ represent the cost of non-gas fired units (non-GFUs) generation and gas source production, respectively.

\subsubsection{Constraints of Power Flow} 
\begin{gather*}
	\sum_{i\in \Omega _{\mathrm{NG}}}{P_{i,t}^{\mathrm{NG}}}+\sum_{i\in \Omega _{\mathrm{G}}}{P_{i,t}^{\mathrm{G}}}+\sum_{i\in \Omega _{\mathrm{WF}}}{P_{i,t}^{\mathrm{PV}}}-\sum_{i\in \Omega _{\mathrm{P}2\mathrm{G}}}{P_{i,t}^{\mathrm{P2H}}}=\sum_{i\in \Omega _{\mathrm{EL}}}{P_{i,t}^{\mathrm{L}}}, \tag{2a}
	\\
	-P_{l}^{\max}\le \sum_{i\in \Omega _{\mathrm{NG}}}\!\!{\mathrm{GSF}}_{l,i}P_{i,t}^{\mathrm{NG}}+\!\sum_{i\in \Omega _{\mathrm{G}}}\!{\mathrm{GSF}}_{l,i}P_{i,t}^{\mathrm{G}}+\!\sum_{i\in \Omega _{\mathrm{WF}}}\!{\mathrm{GSF}}_{l,i}P_{i,t}^{\mathrm{PV}},
	\\
	-\sum_{i\in \Omega _{\mathrm{P}2\mathrm{G}}}\!\!{\mathrm{GSF}}_{l,i}P_{i,t}^{\mathrm{P2H}}-\!\sum_{i\in \Omega _{\mathrm{EL}}}\!\!{\mathrm{GSF}}_{l,i}P_{i,t}^{\mathrm{L}}\le P_{l}^{\max}, \tag{2b}
	\\
	P_{i}^{\mathrm{NG},\min}\le P_{i,t}^{\mathrm{NG}}\le P_{i}^{\mathrm{NG},\max}, \tag{2c}
	\\
	P_{i}^{\mathrm{G},\min}\le P_{i,t}^{\mathrm{G}}\le P_{i}^{\mathrm{G},\max}. \tag{2d}
\end{gather*}

Constraint (2a) ensures the power balance of the grid.
For illustration, suppose the renewable energy source is photovoltaic (PV).
Constraint (2b) limits the power on each line of the grid, where $\mathrm{GSF}_{l,i}$ is the generation shift factor from node $i$ to line $l$ \cite{daniel2019VarianceAware}. 
Constraints (2c) and (2d) set the power limits for non-GFUs and gas fired units (GFUs), respectively.
\vspace{0pt}
\subsubsection{Constraints of Hydrogen Station} 
Based on the composition of hydrogen stations depicted in Fig. \ref{fig_config}, 
the operational constraints can be expressed as follows:
\begin{gather*}
	F_{m,t}^{\mathrm{P2H}}=\beta _{i}^{\mathrm{P2H}}P_{i,t}^{\mathrm{P2H}}, \tag{3a}
	\\
	E_{m,t}=E_{m,t-1}+\eta _{m,t}^{\mathrm{ch}}F_{m,t}^{\mathrm{P2H}}-F_{m,t}^{\mathrm{dch}}/\eta _{m,t}^{\mathrm{dch}}, \tag{3b}
	\\
	E_{m}^{\min}\leq E_{m,t}\leq E_{m}^{\max}, \tag{3c}
\end{gather*} 
\begin{gather*}
	E_{m,0}^{}\leq E_{m,T}, \tag{3d}
	\\
	P_{i}^{\mathrm{P2H},\min}\le P_{i,t}^{\mathrm{P2H}}\le P_{i}^{\mathrm{P2H},\max}, \tag{3e}
	\\
	0 \leq F_{m,t}^{\mathrm{P2H}} \leq F_{m,t}^{\mathrm{ch},\max},\tag{3f}
	\\
	0 \leq F_{m,t}^{\mathrm{dch}} \leq F_{m,t}^{\mathrm{dch},\max}.\tag{3g}
\end{gather*} 

Constraint (3a) specifies the conversion efficiency of P2Hs.
Constraints (3b) and (3c) represent the temporal relationship and capacity limitations of hydrogen storage, respectively.
Constraint (3d) ensures the hydrogen stored at the end of the day is not less than the initial amount, guaranteeing dispatch availability for subsequent days.
Constraints (3e), (3f), and (3g) establish the limits for P2H power, and the charge and discharge rates of hydrogen storage, respectively.
\begin{figure}[!t]
	\centering
	\includegraphics[width=3.2in]{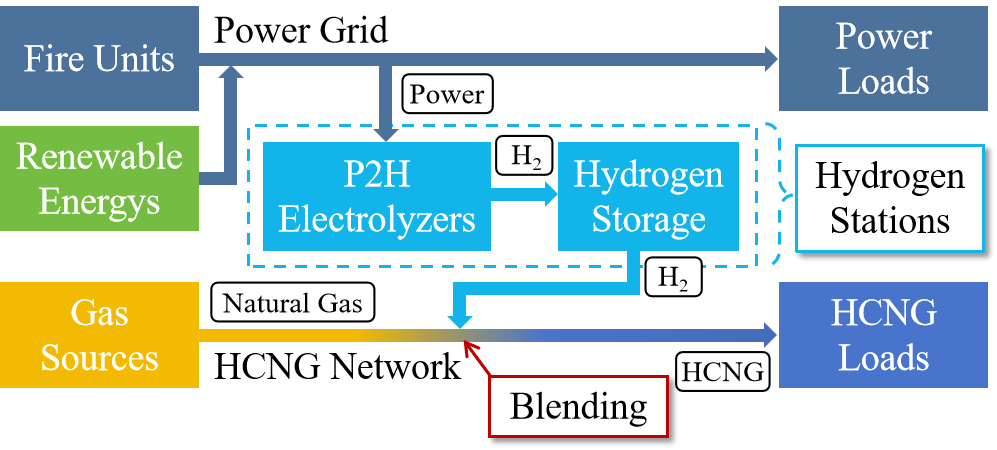}
	\caption{Configuration of E-HCNG network system.}
	\vspace{-12pt}
	\label{fig_config}
\end{figure}

\vspace{-6pt}
\subsubsection{Constraints of HCNG Flow} 
\begin{gather*}
	F_{m,t}^{S}+F_{m,t}^{dch}-F_{m,t}^{G}-F_{m,t}^{L}-\sum_{k\in \Omega _{C,m}}{v_{k}^{C}F_{k,t}^{C}}\\
	=\sum_{n\in \Omega _m}{F_{mn,t}}+\sum_{k\in \Omega _{C,m}}{F_{k,t}^{C}}, \tag{4a}
	\\
	{F_{mn,t}}^{2}= sgn\left( F_{mn,t} \right) C_{mn,t}^{}\left( {\pi _{m,t}}^{2}-{\pi _{n,t}}^{2} \right),  \tag{4b}
	\\
	LP_{mn,t}=K_{mn}\left( \pi _{m,t}+\pi _{n,t} \right) /2,\tag{4c}
	\\
	F_{mn,t}^{\mathrm{in}}-F_{mn,t}^{\mathrm{out}}=LP_{mn,t}-LP_{mn,t-1},\tag{4d}
	\\
	F_{mn,t}=\left( F_{mn,t}^{\mathrm{in}}+F_{mn,t}^{\mathrm{out}} \right) /2,\tag{4e}
	\\
	\sum_{mn\in \Omega _{\mathrm{pipe}}}{L}P_{mn,0} \le \sum_{mn\in \Omega _{\mathrm{pipe}}}{L}P_{mn,T} ,\tag{4f}
	\\
	-F_{mn}^{\max}\le F_{mn,t}\le F_{mn}^{\max}\tag{4g}
	\\
	r_{k}^{\min}\pi _{k,t}^{\mathrm{in}}\le \pi _{k,t}^{\mathrm{out}}\le r_{k}^{\max}\pi _{k,t}^{\mathrm{in}},\tag{4h}
	\\
	0\le F_{k,t}^{C}\le F_{k,t}^{C,\max},\tag{4i}
	\\
	\pi _{m}^{\min}\le \pi _{m,t}\le \pi _{m}^{\max},\tag{4j}
	\\
	F_{m}^{\mathrm{S},\min}\le F_{m,t}^{\mathrm{S}}\le F_{m}^{\mathrm{S},\max},\tag{4k}
	\\
	w_t=\frac{\sum_{m\in \Omega _{\mathrm{P2H}}}{F_{m,t}^{\mathrm{dch}}}+w_{t-1}^{}\sum_{m\in \Omega _{\mathrm{pipe}}}{LP_{m,t-1}}}{\sum_{m\in \Omega _{\mathrm{GS}}}{F_{m,t}^{\mathrm{S}}}+\sum_{m\in \Omega _{\mathrm{P2H}}}{F_{m,t}^{\mathrm{dch}}}+\sum_{m\in \Omega _{\mathrm{pipe}}}{LP_{m,t-1}}},\tag{4l}
	\\
	C_{mn,t}^{}=\left( \frac{\pi}{4} \right) ^2\frac{{D_{ij}}^{5}R{\tau_{n}}^{2}}{\lambda L_{ij}{p_{n}}^{2}M_tZ\tau} ,\tag{4m}
	\\
	M_t=w_tM^{\mathrm{Hy}}+\left( 1-w_t \right) M^{Gas},\tag{4n}
	\\
	F_{m,t}^L=\frac{Heat_{m,t}^L}{HHV_t},\tag{4o}
	\\
	F_{m,t}^G=\frac{\beta^G_{m}P_{m,t}^G}{HHV_t},\tag{4p}
	\\
	HHV_t =w_tHHV^{\mathrm{Hy}}+\left( 1-w_t \right) HHV^{\mathrm{Gas}}.\tag{4q}
\end{gather*}

Formulations (4a)-(4k) describe the general operational constraints of gas flow, accounting for line packs and compressors.
Constraint (4a) ensures node flow balance.
Constraint (4b) describes the relationship between line flow and nodal pressure \cite{m.2016Steady}, where the sign function ($sgn$) indicates the direction of flow.
Constraint (4c) details the interaction between gas storage in line packs and nodal pressure.
Constraint (4d) describes the temporal relationship between line pack storage and flow rates.
Constraint (4e) calculates the average gas flow in pipelines.
Constraint (4f) restricts the line pack storage within a day horizon, similar to constraint (3d).
Constraint (4g) sets the maximum permissible gas flow rate.
Constraints (4h) and (4i) limit the output pressure and flow rate of gas compressors, respectively.
Constraints (4j) and (4k) establish the limits for nodal pressure and gas source rate.

Formulations (4l)-(4p) describe the unique constraints for HCNG flow introduced by the variable HVF.
Constraint (4l) estimates the HVF considering line pack storages \cite{yunpeng2024ChongDa}.
Constraints (4m) and (4n) reflect how HVF influences the parameters in the HCNG flow equation.
Specifically, in (4m), $D_{ij}$, $\lambda_{ij}$, $L_{ij}$ represent the diameter, friction coefficient, and length of pipeline $ij$, respectively.
$R$ denotes the gas constant, $\tau_n$ and $p_{n}$ are the standard temperature and pressure, $M$ is the gas molar mass, $Z$ is the compressibility factor, and $\tau$ represents the gas temperature.
Constraints (4o) and (4p) calculate the required HCNG volume by load and GFUs, respectively, based on the heat value calculated through (4q) \cite{chen2020GuanDaoXianQing1}.

\vspace{-10pt}
\subsection{Affine Policy for the E-HCNG System}
In actual operations, the rapid volatility of renewable energy threatens the energy balance and economic efficiency of the E-HCNG network operations.
To address this issue, a system-level affine policy is established for E-HCNG networks.
The affine policy can achieve real-time re-dispatch of operations based on the discrepancy between predicted and actual renewable energy outputs.


For clarity, let $\widetilde{\left( \cdot \right) }$ denote the actual execution of operation $\left( \cdot \right)$ after re-dispatch. Let $\boldsymbol{\xi}_t$ represent the discrepancy between the predicted and actual PV outputs at time $t$, i.e., $\tilde{P}_{i,t}^{\mathrm{PV}}=P_{i,t}^{\mathrm{PV}}+\xi _{i,t}$, and $\boldsymbol{\xi}_t=\{\xi _{i,t}\}$. 


\subsubsection{Affine Policy for the Grid}

Under the affine policy, the actual outputs of non-GFUs, GFUs, and P2Hs can be expressed as
\begin{align*}
	\tilde{P}_{i,t}^{\mathrm{NG}}&=P_{i,t}^{\mathrm{NG}}-\alpha _{i,t}^{\mathrm{NG}} \boldsymbol{e}^{\mathrm{T}}\boldsymbol{\xi} _t , \tag{5a}
	\\
	\tilde{P}_{i,t}^{\mathrm{G}}&=P_{i,t}^{\mathrm{G}}-\alpha _{i,t}^{\mathrm{G}} \boldsymbol{e}^{\mathrm{T}}\boldsymbol{\xi} _t , \tag{5b}
	\\
	\tilde{P}_{i,t}^{\mathrm{P2H}}&=P_{i,t}^{\mathrm{P2H}}+\alpha _{i,t}^{\mathrm{P2H}} \boldsymbol{e}^{\mathrm{T}}\boldsymbol{\xi} _t , \tag{5c}
\end{align*}
where $\alpha_{i,t}^{(\cdot)}$ represents the corresponding adjust factor of the affine policy, and $\boldsymbol{e}$ is a vector of all ones. 
Moreover, substituting (5a)-(5c) into (2a) results in a derived constraint for the adjustment factors as follow, which ensures the energy balance during adjustment.
\begin{align*}
	\sum_{i\in \Omega _{\mathrm{NG}}}{\alpha _{i,t}^{\mathrm{NG}}}+\sum_{i\in \Omega _{\mathrm{G}}}{\alpha _{i,t}^{\mathrm{G}}}+\sum_{i\in \Omega _{\mathrm{P}2\mathrm{G}}}{\alpha _{i,t}^{\mathrm{P2H(P)}}}=1. \tag{6}
\end{align*}

\subsubsection{Affine Policy for Hydrogen Components}

For P2H and hydrogen storage, the operation under affine policy can be expressed as
\begin{align*}
\tilde{F}_{m,t}^{\mathrm{P2H}}&=F_{m,t}^{\mathrm{P2H}}+\alpha _{m,t}^{\mathrm{P2H(F)}} \boldsymbol{e}^{\mathrm{T}}\boldsymbol{\xi} _t ,  \tag{7a}
\\
\alpha_{i,t}^{\mathrm{P2H(F)}}&=\beta _{i}^{\mathrm{P2H}}\alpha_{i,t}^{\mathrm{P2H(P)}} ,\tag{7b}
\\
\tilde{F}_{m,t}^{\mathrm{dch}}&=F_{m,t}^{\mathrm{dch}}+\alpha _{m,t}^{\mathrm{dch}} \boldsymbol{e}^{\mathrm{T}}\boldsymbol{\xi} _t   ,\tag{7c}
\\
\tilde{E}_{i,t} &= E_{i,t} + \alpha_{i,t}^{E} \boldsymbol{e}^{\mathrm{T}}\boldsymbol{\xi} _t  ,\tag{7d}
\\
\alpha_{i,t}^{E} &= \alpha_{i,t-1}^{E} + \eta _{i}^{\mathrm{ch}}\alpha_{m,t}^{\mathrm{P2H(F)}} - \alpha _{m,t}^{\mathrm{dch}}/\eta _{i}^{\mathrm{dch}} .\tag{7e}
\end{align*}

Equation (7a) represents the actual P2H output. 
By substituting (7a) and (5c) into (3a), the relationship between the adjustment factors for power input and hydrogen output in the P2H can be obtained in (7b).
Equations (7c) and (7d) specify the actual amounts of hydrogen discharged and stored, respectively. 
Equation (7e) describes how the adjustment factors relate the hydrogen storage capacities to both the P2H output and the storage discharge.

\subsubsection{Affine Policy for the HCNG Network}
GFUs, P2Hs, and hydrogen storages link the grid to HCNG networks.
Thus, by fine-tuning the HCNG networks, the system can more effectively mitigate the impacts of renewable energy fluctuations, ultimately improving operational economic efficiency.
\begin{align*}
	\tilde{F}_{m,t}^{\mathrm{S}}&=F_{m,t}^{\mathrm{S}}+\alpha _{m,t}^{\mathrm{S}} \boldsymbol{e}^{\mathrm{T}}\boldsymbol{\xi} _t, \tag{8a}
	\\
	\tilde{F}_{k,t}^{\mathrm{C}}&=F_{k,t}^{\mathrm{C}}+\alpha _{k,t}^{\mathrm{C}} \boldsymbol{e}^{\mathrm{T}}\boldsymbol{\xi} _t, \tag{8b}
	\\
	\tilde{\pi}_{m,t}&=\pi _{m,t}+\alpha _{m,t}^\pi \boldsymbol{e}^{\mathrm{T}}\boldsymbol{\xi} _t,  \tag{8c}
	\\
	\tilde{F}_{mn,t}^{\mathrm{in}}&=F_{mn,t}^{\mathrm{in}}+\alpha _{mn,t}^{\mathrm{in}} \boldsymbol{e}^{\mathrm{T}}\boldsymbol{\xi} _t,  \tag{8d}
	\\
	\tilde{F}_{mn,t}^{\mathrm{out}}&=F_{mn,t}^{\mathrm{out}}+\alpha _{mn,t}^{\mathrm{out}} \boldsymbol{e}^{\mathrm{T}}\boldsymbol{\xi} _t,  \tag{8e}
	\\
	\tilde{F}_{mn,t}&=F_{mn,t}+\alpha _{mn,t} \boldsymbol{e}^{\mathrm{T}}\boldsymbol{\xi} _t , \tag{8f}
	\\
		\tilde{F}_{m,t}^{\mathrm{G}}&=F_{m,t}^{\mathrm{G}}+\alpha _{m,t}^{\mathrm{G}} \boldsymbol{e}^{\mathrm{T}}\boldsymbol{\xi} _t . \tag{8g}
\end{align*}

\vspace{-3pt}
Equations (8a)-(8f) represent the actual operations of gas source, nodal pressure, pipeline inflow, pipeline outflow, pipeline average flow rates, and GFUs, respectively.
Substituting these equations into (4a)-(4e) yields the relationships among the adjustment factors within the HCNG network as follows:
\begin{align*}
	\alpha _{m,t}^{\mathrm{S}}+\alpha _{m,t}^{\mathrm{dch}}&-\alpha _{m,t}^{\mathrm{G}}-\sum_{k\in \Omega _{\mathrm{C,m}}}{v_{k}^{\mathrm{C}}\alpha _{k,t}^{\mathrm{C}}} ,
	\\
	&=\sum_{n\in \Omega _m}{\alpha _{mn,t}}+\sum_{k\in \Omega _{\mathrm{C,m}}}{\alpha _{k,t}^{\mathrm{C}}},\tag{9a}
	\\
	{\alpha _{mn,t}}^{2}&=\mathrm{sign}\left( F_{mn,t} \right) C_{mn,t}\left( ({\alpha _{m,t}^\pi})^{2}-{(\alpha _{n,t}^\pi)}^{2} \right),\tag{9b}
	\\
	\alpha _{mn,t}F_{mn,t}&=\mathrm{sign}\left( F_{mn,t} \right) C_{mn,t}^{}\left( \alpha _{m,t}^\pi \pi_{m,t}-\alpha _{n,t}^\pi\pi _{n,t} \right)\!,\!\!\!\!\tag{9c}
	\\
	\alpha _{mn,t}^{\mathrm{in}}-\alpha _{mn,t}^{\mathrm{out}}&=\frac{K_{mn}}{2}\left( \alpha^\pi _{m,t}+\alpha^\pi _{n,t}-\alpha^\pi_{m,t-1}-\alpha^\pi_{n,t-1} \right),\tag{9d} 
	\\
	\alpha _{mn,t}&=\left( \alpha _{mn,t}^{\mathrm{in}}+\alpha _{mn,t}^{\mathrm{out}} \right) /2.\tag{9e}
\end{align*}

\vspace{-3pt}
Equation (9a) ensures the nodal flow balance of adjustment factors.
Equations (9b) and (9c) describe the adjustment relationship between line flows and nodal pressures.
Equation (9d) represents the line pack adjustments influenced by changes in nodal pressure.
Equation (9e) defines the relationship of adjustment factors between average flow and inflow/outflow rates.
Furthermore, it is assumed that HVF remains unaffected by these adjustments.
This stability is attributed to the adjustments being relatively minor compared to the original scheduling, and is further supported by the buffering effect of line packs.

\vspace{-10pt}
\subsection{Formulation of DRJCC E-HCNG Scheduling}
Based on the above affine policy, operations are adjusted with respect to $\boldsymbol{\xi}_t$. 
However, these operations should remain within the operational constraints.
Therefore, to ensure reliability, DRJCC is introduced to limit the violation probability of operation constraints, under the uncertainties of renewable energy.
In this subsection, first, a distributionally robust ambiguity set is developed to characterize the uncertainties of PV output discrepancy $\boldsymbol{\xi}$. 
Subsequently, DRJCC and a unified operational formulation are established.


\subsubsection{Ambiguity Set for PV Output Uncertainty}

To more precisely characterize the uncertainty of the PV output discrepancy $\boldsymbol{\xi}_t$, an unimodality-skewness informed second-order moment-based ambiguity set is introduced \cite{chao2022LiLun}.
\begin{align*}
\!\!\!\!\mathcal{D} =\left\{ \mathbb{P} _{\xi}\!\in\! \mathcal{P(\mathcal{A})}\! \left| \!\begin{array}{l}
	\left[ \mathbb{E} _{\mathbb{P}}(\boldsymbol{\xi })-\boldsymbol{\mu }_0 \right] ^{\mathrm{T}}\mathbf{\Sigma }_{0}^{-1}\left[ \mathbb{E} _{\mathbb{P}}(\boldsymbol{\xi })-\boldsymbol{\mu }_0 \right] \le \gamma _1,\\
	\mathbb{E} _{\mathbb{P}}\left[ \left( \boldsymbol{\xi }-\boldsymbol{\mu }_0 \right) \left( \boldsymbol{\xi }-\boldsymbol{\mu }_0 \right) ^{\mathrm{T}} \right] \preceq \gamma _2\mathbf{\Sigma }_0,\\
	\mathbb{P} _{\xi}\,\,\mathrm{is}\,\,\alpha \-\mathrm{unimodal}\,\,\mathrm{about}\,\,0.\\
\end{array} \right. \!\!\!\!\! \right\} \!,\!  \tag{10}
\end{align*}
where $\boldsymbol{\xi}=\{\boldsymbol{\xi}_t, \forall t \in T\}$. 
$\mathbb{P}_\xi$ represents the probability distribution of $\boldsymbol{\xi}$. $\mathcal{P(\mathcal{A})}$ is the set of all distributions of $\boldsymbol{\xi}$ on the sigma-field of support set $\mathcal{A}$.
$\boldsymbol{\mu}_0$ and $\mathbf{\Sigma }_{0}$ are the estimated first- and second-order moment of $\boldsymbol{\xi}$.
$\gamma_1$ and $\gamma_2$ restrict the Isogai's mode skewness and the deviation of the second-order moment from $\mathbf{\Sigma }_{0}$\cite{chao2022LiLun}.
The values of $\gamma_1 \geq 0$ and $\gamma_2 \geq 1$ can be typically determined through cross-validation in practice \cite{yiling2018CC_solve}.

The $\alpha$-unimodal is defined as: If $\boldsymbol{\xi}$ admits a density $f_\xi(x)$, then $\mathbb{P}_\xi$ is considered $\alpha$-unimodal about 0 if and only if the generalized univariate density $\tilde{f}{\alpha,x}(t) = t^{d-\alpha} f_\xi (tx)$ is nonincreasing for $t > 0$ across all $z \in \mathbb{R}^d$ \cite{Olshen_Savage_1970}.
The unimodality defines the degree of aggregation of PV prediction errors, reducing over-consideration for rarely deviated samples. 
This property is suitable for uncertainties of renewable energy output prediction errors \cite{chao2022LiLun}.

\subsubsection{Formulation of the DRJCC}

Under the affine policy, operations after adjustment should 
comply with both component limits and reserve capacities.
As for the component limits, the inequality constraints from the deterministic model are modified to incorporate these adjustments effectively. 
For clarity, consider constraints (2b) and (2c) as examples.
Their modified constraints $\widetilde{(2\mathrm{b)} }$ and $\widetilde{(2\mathrm{c)} }$ under the affine policy can be reformulated as
\vspace{-5pt}
\begin{gather*}
	-P_{l}^{\max}\le \!\sum_{i\in \Omega _{\mathrm{NG}}}\!{G}SF_{l,i}\tilde{P}_{i,t}^{\mathrm{NG}}+\!\sum_{i\in \Omega _{\mathrm{G}}}\!{G}SF_{l,i}\tilde{P}_{i,t}^{\mathrm{G}}+\!\sum_{i\in \Omega _{\mathrm{WF}}}\!\!{G}SF_{l,i}\tilde{P}_{i,t}^{\mathrm{PV}} 
	\\
	-\sum_{i\in \Omega _{\mathrm{P}2\mathrm{G}}}\!{G}SF_{l,i}\tilde{P}_{i,t}^{\mathrm{P2H}}-\sum_{i\in \Omega _{\mathrm{EL}}}\!{G}SF_{l,i}\tilde{P}_{i,t}^{\mathrm{L}}\le P_{l}^{\max} 
	,\tag{$\widetilde{2\mathrm{b} }$}
	\\
	P_{i}^{\mathrm{NG},\min}\le \tilde{P}_{i,t}^{\mathrm{NG}}\le P_{i}^{\mathrm{NG},\max}  \tag{$\widetilde{2\mathrm{c} }$},
\end{gather*}

\vspace{-5pt}
\noindent where $\tilde{P}_{i,t}^{(\cdot)}$ represents the corresponding operations of ${P}_{i,t}^{(\cdot)}$ under affine policy, as detailed in Section II-B.
Similarly, other inequality constraints under the affine policy including $\widetilde{(2\mathrm{d)}},\widetilde{(3\mathrm{c)}}-\widetilde{(3g)},\widetilde{(4f)}-\widetilde{(4k)}$ are adjusted in the same way.

In addition, the adjustments made by affine policy should comply with the reserve capacities, i.e.
\vspace{-5pt}
\begin{align*}
R_{i,t}^{\mathrm{NG},\mathrm{D}}&\le -\alpha _{i,t}^{\mathrm{NG}}\boldsymbol{e}^{\mathrm{T}}\boldsymbol{\xi }_t\le R_{i,t}^{\mathrm{NG},\mathrm{U}} ,\tag{11a}
\\
R_{i,t}^{\mathrm{G},\mathrm{D}}&\le -\alpha _{i,t}^{\mathrm{G}}\boldsymbol{e}^{\mathrm{T}}\boldsymbol{\xi }_t\le R_{i,t}^{\mathrm{G},\mathrm{U}} ,\tag{11b}
\\
R_{i,t}^{\mathrm{P}2\mathrm{H},\mathrm{D}}&\le \alpha _{i,t}^{\mathrm{P}2\mathrm{H}}\boldsymbol{e}^{\mathrm{T}}\boldsymbol{\xi }_t\le R_{i,t}^{\mathrm{P}2\mathrm{H},\mathrm{U}} ,\tag{11c}
\\
R_{m,t}^{\mathrm{S},\mathrm{D}}&\le \alpha _{i,t}^{\mathrm{S}}\boldsymbol{e}^{\mathrm{T}}\boldsymbol{\xi }_t\le R_{m,t}^{\mathrm{S},\mathrm{U}},\tag{11d}
\end{align*}
where variables $R_{i,t}^{(\cdot),\mathrm{U}}$ and $R_{i,t}^{(\cdot),\mathrm{D}}$ denote the corresponding upper and lower reserve capacities of $(\cdot)$, which should also be determined in day-ahead.

To ensure reliability, DRJCC is implemented to limit the violation probability under the uncertainties of PV output prediction discrepancy.
The DRJCC can be expressed as
\vspace{-5pt}
\begin{gather*}
\underset{\mathbb{P} _{\xi}\in \mathcal{D}}{\mathrm{inf}}\,\,\mathbb{P} _{\boldsymbol{\xi }}\left\{ \begin{array}{c}
	\widetilde{(2\mathrm{b)}}-\widetilde{(2\mathrm{d)}},\widetilde{(3\mathrm{c)}}-\widetilde{(3g)},\\
	\widetilde{(4f)}-\widetilde{(4k)},\left( 11a \right) -\left( 11d \right)\\
\end{array} \right\} \ge 1-\epsilon,   \tag{12}
\end{gather*}
where $\mathbb{P} _{\boldsymbol{\xi}}$ denotes the probably distribution of $\boldsymbol{\xi}$ within the ambiguity set $\mathcal{D}$ defined in (10).
The parameter $\epsilon$ is the overall violation probability, also known as the violation risk, of the DRJCC.
$\widetilde{(2\mathrm{b)}}\!-\!\widetilde{(2\mathrm{d)}}$, $\widetilde{(3\mathrm{c)}}\!-\!\widetilde{(3\mathrm{g)}}$, and $\widetilde{(4\mathrm{f)}}\!-\!\widetilde{(4\mathrm{k)}}$ represent inequality constraints of grid, hydrogen stations, and HCNG network, under affine policy, respectively.

\subsubsection{The Unified Operating Formulation for E-HCNG Networks}
Integrating the above formulations, the unified DRJCC scheduling for the E-HCNG network under the affine policy can be expressed as:
\vspace{-5pt}
\begin{align*}
	&\underset{\Phi }{\min}\,\,F:= F^0 +\sum_{t=1}^T{\sum_{i\in \Omega _{\mathrm{G}}}{\left( R_{i,t}^{\mathrm{G},\mathrm{U}}c_{i,t}^{\mathrm{G},\mathrm{U}}+R_{i,t}^{\mathrm{G},\mathrm{D}}c_{i,t}^{\mathrm{G},\mathrm{D}} \right)}}
	\\
	&\qquad\qquad+\sum_{t=1}^T{\sum_{i\in \Omega _{\mathrm{NG}}}{\left( R_{i,t}^{\mathrm{NG},\mathrm{U}}c_{i,t}^{\mathrm{NG},\mathrm{U}}+R_{i,t}^{\mathrm{NG},\mathrm{D}}c_{i,t}^{\mathrm{NG},\mathrm{D}} \right)}}
	\\
	&\qquad\qquad+\sum_{t=1}^T{\sum_{i\in \Omega _{\mathrm{P2H}}}{\left( R_{i,t}^{\mathrm{P}2\mathrm{H},\mathrm{U}}c_{i,t}^{\mathrm{P}2\mathrm{H},\mathrm{U}}+R_{i,t}^{\mathrm{P}2\mathrm{H},\mathrm{D}}c_{i,t}^{\mathrm{P}2\mathrm{H},\mathrm{D}} \right)}}
	\\
	&\qquad\qquad+\sum_{t=1}^T{\sum_{m\in \Omega _{\mathrm{S}}}{\left( R_{m,t}^{\mathrm{S},\mathrm{U}}c_{m,t}^{\mathrm{S},\mathrm{U}}+R_{m,t}^{\mathrm{S},\mathrm{D}}c_{m,t}^{\mathrm{S},\mathrm{D}} \right)}}\tag{13}
\\
	&\,\mathrm{s}.\mathrm{t}.\;\,\left( 2a \right) ,\left( 3a \right) ,\left( 3b \right) ,\left( 4a \right) -\left( 4e \right) ,\left( 4l \right) -\left( 4q \right) ,\left( 5a \right) -\left( 5c \right) ,\\
	&\qquad    \left( 6 \right) ,\left( 7a \right) -\left( 7e \right) ,\left( 8a \right) -\left( 8g \right) ,\left( 9a \right) -\left( 9e \right) ,\left( 12 \right),\\
\end{align*}

\vspace{-15pt}
\noindent where $\Phi $ is the set of all scheduling decisions including $\Phi ^0$ and all affine factors. 
As for object $F(\Phi)$, the first two terms consistent with the deterministic model (1),
and the third to sixth terms quantify the reserve capacity costs of GFUs, non-GFUs, P2Hs, and gas sources, respectively.
Parameters $c_{i,t}^{(\cdot),\mathrm{U}}$ and $c_{i,t}^{(\cdot),\mathrm{D}}$ are the corresponding per unit-costs of reserve capacities.
Note that the DRJCC (12) and the high-order non-convexity introduced by HVF are intractable. These will be addressed in the next section. 


\vspace{-5pt}
\section{Solutions}

In this section, the DRJCC is first separated into individual DRCCs through an improved risk allocation method that utilizes the correlation of violations under the affine policy. 
Subsequently, the individual DRCCs are transformed into second-order cone (SOC) constraints.
Finally, customized approximations are adopted to effectively handle the non-convexities in the HCNG network.
\vspace{-10pt}
\subsection{Reformulation of DRJCCs}

For ease of illustration, the compact form of DRJCC (12) can be expressed as:
\begin{align*}
	&\underset{\mathbb{P} _{\boldsymbol{\xi}}\in \mathcal{D}}{\mathrm{inf}}\,\,\mathbb{P} _{\boldsymbol{\xi }}\left\{ a_i\left( \boldsymbol{x} \right) \boldsymbol{e}^T\boldsymbol{\xi }\le b_i\left( \boldsymbol{x} \right) ,\forall i\in \left[ N \right] \right\} \ge 1-\epsilon , \tag{14}
\end{align*}
where  $a_i\left( \boldsymbol{x} \right)$ and $b_i\left( \boldsymbol{x} \right)$ are both linear with respect to the decision vector $\boldsymbol{x}$.
$\left[ N \right]=\{1,2,...,N\}$ enumerates the individual operational conditions in the DRJCC. 

To tackle DRJCC, a major problem is rational allocating the violating probabilities $\epsilon$ to each constraint, so transforming the DRJCC into individual DRCCs.
\subsubsection{Existing Bonferroni Approximation Approaches for DRJCC}
 The most commonly used \textit{averaged} Bonferroni approximation transforms the DRJCC (14) as: 
 \begin{align*}
 	& \inf _{\mathbf{P} \in \mathcal{P}} \mathbb{P}\left\{a_i(x) \boldsymbol{e}^T\xi \leq b_i(x)\right\} \geq 1-\frac{\epsilon}{N}, \forall i \in[N]. \tag{15}
 \end{align*}
 This method can easily achieve risk allocation of $\epsilon$, without increasing the complexity.
 
 Latter on, to realize a more rational risk allocation, the \textit{optimized} Bonferroni approximation was proposed, and (14) is transformed into
\begin{align*}
	& \inf _{\mathbf{P} \in \mathcal{P}} \mathbb{P}\left\{a_i(x) \boldsymbol{e}^T\xi \leq b_i(x)\right\} \geq 1-\epsilon_i, \forall i \in[N], \tag{16a}
	\\
	& \sum_{i \in[N]} \epsilon_i \leq \epsilon, \epsilon_i \geq 0, \forall i \in[N],\tag{16b}
\end{align*}

\vspace{-5pt}
\noindent where $\epsilon_i$ are recognized as decision variables.
Correspondingly, an iterative solving method was proposed to get optimal $\epsilon_i$ \cite{lun2022YangLun}.

Although these methods have their own advantages, unfortunately, 
they do not suit the DRJCC that includes numerous conditions, such as E-HCNG networks.
When $N$ is large, the averaged Bonferroni approximation may result in extreme over-conservation due to the small ${\epsilon}/{N}$. And the optimized Bonferroni approximation may encounter a computational burden.

Furthermore, the Bonferroni approximation-based method potentially assumes that the violations of each condition $a_i(x) \xi \le b_i(x)$ are probabilistically mutually exclusive \cite{weijun2022LiLun}.
However, the conditions in DRJCC are highly correlated under the affine policy.
This inspired us to propose an improved transformation for DRJCC with the affine policy.

\subsubsection{Improved Transformation for DRJCC with Affine Policy}


For the condition such as
$	a_i\left( \boldsymbol{x} \right) ^T\boldsymbol{\xi }\le b_i\left( \boldsymbol{x} \right)$,
represent $a_i\left( \boldsymbol{x} \right)$ as
\begin{gather}
	a_i\left( \boldsymbol{x} \right) =\left( a_{i}^{+}\left( \boldsymbol{x} \right) +\delta \right) -\left( a_{i}^{-}\left( \boldsymbol{x} \right) +\delta \right) ,
	\tag{17a}
	\\
	a_{i}^{+}\left( \boldsymbol{x} \right) \cdot a_{i}^{-}\left( \boldsymbol{x} \right) =0,
	\tag{17b}
	\\
	a_{i}^{+}\left( \boldsymbol{x} \right),a_{i}^{-}\left( \boldsymbol{x} \right) >0, \tag{17c}
	\\
	\delta \rightarrow 0^+.
	\tag{17d}
\end{gather}
The complementary condition (17b) can be linearized using the big-M method \cite{jingxiang2024V2G}.

Then, due to $a_{i}^{+}\left( \boldsymbol{x} \right) +\delta >0$ and $a_{i}^{-}\left( \boldsymbol{x} \right) +\delta>0$,
the original conditions in DRJCC (14) are equal to 
\begin{align*}
	a_{i}^{}\left( \boldsymbol{x} \right) \boldsymbol{e}^T\boldsymbol{\xi} \le b_i\left( \boldsymbol{x} \right) 
	\Leftrightarrow \begin{cases}
		\boldsymbol{e}^T\boldsymbol{\xi} \le \frac{b_i\left( \boldsymbol{x} \right)}{a_{i}^{+}\left( \boldsymbol{x} \right) +\delta}\\
		\boldsymbol{e}^T\boldsymbol{\xi} \ge -\frac{b_i\left( \boldsymbol{x} \right)}{a_{i}^{-}\left( \boldsymbol{x} \right) +\delta}.\tag{18}
	\end{cases}
\end{align*}

Therefore, the DRJCC (14) can be rewritten as
\begin{align*}
\underset{\mathbb{P} _{\xi}\in \mathcal{D}}{\mathrm{inf}}\,\,\mathbb{P} _{\boldsymbol{\xi }}\left\{ \begin{array}{c}
	\boldsymbol{e}^T\boldsymbol{\xi }\le \frac{b_i\left( \boldsymbol{x} \right)}{a_{i}^{+}\left( \boldsymbol{x} \right) +\delta}\\
	\boldsymbol{e}^T\boldsymbol{\xi }\ge -\frac{b_i\left( \boldsymbol{x} \right)}{a_{i}^{-}\left( \boldsymbol{x} \right) +\delta}\\
\end{array},\forall i\in \left[ N \right] \right\} \ge 1-\epsilon 
\tag{19a}
\\
\Leftrightarrow \underset{\mathbb{P} _{\xi}\in \mathcal{D}}{\mathrm{inf}}\,\,\mathbb{P} _{\boldsymbol{\xi }}\left\{ \begin{array}{c}
	\boldsymbol{e}^T\boldsymbol{\xi }\le \underset{i\in [N]}{\inf}\left\{ \frac{b_i\left( \boldsymbol{x} \right)}{a_{i}^{+}\left( \boldsymbol{x} \right) +\delta} \right\}\\
	\boldsymbol{e}^T\boldsymbol{\xi }\ge \underset{i\in [N]}{\sup}\left\{ -\frac{b_i\left( \boldsymbol{x} \right)}{a_{i}^{-}\left( \boldsymbol{x} \right) +\delta} \right\}\\
\end{array} \right\} \ge 1-\epsilon . \tag{19b}
\end{align*}

Using Bonferroni approximation on the two inequalities in (19b), can get that 
\vspace{-5pt}
\begin{align*}
	&\underset{\mathbb{P} _{\xi}\in \mathcal{D}}{\mathrm{inf}}\,\,\mathbb{P} _{\boldsymbol{\xi }}\left\{ \boldsymbol{e}^T\boldsymbol{\xi }\le \underset{i}{\inf}\left\{ \frac{b_i\left( \boldsymbol{x} \right)}{a_{i}^{+}\left( \boldsymbol{x} \right) +\delta} \right\} \right\} \ge 1-\frac{\epsilon}{2},\, \tag{20a}
	\\
	&\underset{\mathbb{P} _{\xi}\in \mathcal{D}}{\mathrm{inf}}\,\,\mathbb{P} _{\boldsymbol{\xi }}\left\{ \boldsymbol{e}^T\boldsymbol{\xi }\ge \underset{i}{\sup}\left\{ -\frac{b_i\left( \boldsymbol{x} \right)}{a_{i}^{-}\left( \boldsymbol{x} \right) +\delta} \right\} \right\} \ge 1-\frac{\epsilon}{2}.\,
	\tag{20b}
\end{align*}

\begin{proposition}
For chance constraints (21a) and (21b), 
the following equivalent relationship holds
\begin{align*}
	&\underset{\mathbb{P} _{\xi}\in \mathcal{D}}{\mathrm{inf}}\,\,\mathbb{P} _{\boldsymbol{\xi }}\left\{ \boldsymbol{e}^T\boldsymbol{\xi }\le \underset{i\in[N]}{\inf}\left\{ \frac{b_i\left( \boldsymbol{x} \right)}{a_{i}^{+}\left( \boldsymbol{x} \right) +\delta} \right\} \right\} \ge 1-\frac{\epsilon}{2}\,\, 
	\\
	&\Leftrightarrow \underset{\mathbb{P} _{\xi}\in \mathcal{D}}{\mathrm{inf}}\,\,\mathbb{P} _{\boldsymbol{\xi }}\left\{ \boldsymbol{e}^T\boldsymbol{\xi }\le \frac{b_i\left( \boldsymbol{x} \right)}{a_{i}^{+}\left( \boldsymbol{x} \right) +\delta} \right\} \ge 1-\frac{\epsilon}{2}, \forall i\in \left[ N \right] ;
	\tag{21a}
	\\
	&\underset{\mathbb{P} _{\xi}\in \mathcal{D}}{\mathrm{inf}}\,\,\mathbb{P} _{\boldsymbol{\xi }}\left\{ \boldsymbol{e}^T\boldsymbol{\xi }\ge \underset{i\in[N]}{\sup}\left\{ -\frac{b_i\left( \boldsymbol{x} \right)}{a_{i}^{-}\left( \boldsymbol{x} \right) +\delta} \right\} \right\} \ge 1-\frac{\epsilon}{2}\,\, 
	\\
	&\Leftrightarrow \underset{\mathbb{P} _{\xi}\in \mathcal{D}}{\mathrm{inf}}\,\,\mathbb{P} _{\boldsymbol{\xi }}\left\{ \boldsymbol{e}^T\boldsymbol{\xi }\ge -\frac{b_i\left( \boldsymbol{x} \right)}{a_{i}^{-}\left( \boldsymbol{x} \right) +\delta} \right\} \ge 1-\frac{\epsilon}{2}, \forall i\in \left[ N \right] ;
	\tag{21b}
\end{align*}
\end{proposition}

\begin{proof}
	First, we prove that the condition (20a) equals (21a).
Denote the probability bound function under ambiguity set $\mathcal{D}$ as 
\begin{align*}
	f_{\mathcal{D}}\left( v \right) :=\underset{\mathbb{P} _{\boldsymbol{\xi }}\in \mathcal{D}}{\mathrm{inf}}\,\,\mathbb{P} _{\boldsymbol{\xi }}\left\{ \boldsymbol{e}^T\boldsymbol{\xi }\le v \right\} ,\tag{22}
\end{align*}
where $f_{\mathcal{D}}\left( v \right)$ is monotonically increasing with respect to $v \in \mathbb{R}$ \cite{bartp.g.2019Distributionally}. 
Correspondingly, the inverse function of $f_{\mathcal{D}}\left( v \right)$ can be expressed as 
\begin{align*}
	f_{\mathcal{D}}^{-1}\left( 1-\epsilon \right) =\underset{v}{\mathrm{sup}}\left\{ f_{\mathcal{D}}\left( v \right) \ge 1-\epsilon \right\} .
\end{align*}

Through the definition of $f_{\mathcal{D}}\left( v \right)$,
constraints (20a) and (21a) can be rewritten as follows, respectively
\begin{align*}
	f_{\mathcal{D}}\left( \underset{i}{\mathrm{inf}}\left\{ \frac{b_i\left( \boldsymbol{x} \right)}{a_{i}^{+}\left( \boldsymbol{x} \right) +\delta} \right\} \right) \ge 1-\frac{\epsilon}{2}, \tag{23a}
	\\
	f_{\mathcal{D}}\left( \frac{b_i\left( \boldsymbol{x} \right)}{a_{i}^{+}\left( \boldsymbol{x} \right) +\delta} \right) \ge 1-\frac{\epsilon}{2},\forall i\in \left[ N \right].  \tag{23b}
\end{align*}

Denote $F_0$ as the value that $f_{\mathcal{D}}^{-1}\left( 1-\frac{\epsilon}{2} \right) = F_0$, i.e., $f_{\mathcal{D}}\left( F_0 \right) =1-\frac{\epsilon}{2}$. 
Given that $f_{\mathcal{D}}\left( v \right)$ is monotonically increasing, the (24a) and (24b) holds \textit{if and only if}
\begin{align*}
	\left( 23a \right) \Leftrightarrow \underset{i\in \left[ N \right]}{\mathrm{inf}}\left\{ \frac{b_i\left( \boldsymbol{x} \right)}{a_{i}^{+}\left( \boldsymbol{x} \right) +\delta} \right\} \ge F_0; \tag{24a}
	\\
	\left( 23b \right) \Leftrightarrow \frac{b_i\left( \boldsymbol{x} \right)}{a_{i}^{+}\left( \boldsymbol{x} \right) +\delta}\ge F_0, \forall i\in \left[ N \right] .\tag{24b}
\end{align*}

It is obvious that $\mathrm{(24a)}\Leftrightarrow \mathrm{(24b)}$, so there is $\mathrm{(20a)}\Leftrightarrow \mathrm{(23a)}\Leftrightarrow \mathrm{(24a)}\Leftrightarrow \mathrm{(24b)}\Leftrightarrow \mathrm{(23b)}\Leftrightarrow \mathrm{(21a)}$.
It follows that condition (20a) equals to (21a).

By defining another probability boundary function as 
\begin{align*}
	f_{\mathcal{D}}^{'}\left( v \right) :=\underset{\mathbb{P} _{\boldsymbol{\xi }}\in \mathcal{D}}{\mathrm{inf}}\,\,\mathbb{P} _{\boldsymbol{\xi }}\left\{ \boldsymbol{e}^T\boldsymbol{\xi }\ge v \right\} ,
\end{align*}
which is monotonically decreasing with respect to $v$. 
Using the similar process from (22) to (24), it can be proved that (20b) equals (21b).

\end{proof}
\vspace{-10pt}
\begin{remark}
	(21a) and (21b) can be recognized as the outer approximation of (20a) and (20b), respectively. 
Generally, the outer approximation directly assigns the violation probability to each individual DRCCs, which can not ensure the risk level of the original DRJCC \cite{weijun2022LiLun}.
However, \textit{Proposition 1} demonstrates that for the one-sided inequalities of $\boldsymbol{\xi}$ like (21a) and (21b), the outer approximation is equivalent.
This equivalence is because all operational re-dispatches under the affine policy follow consistent uncertainties, leading to significant correlation. Consequently, only the most restrictive condition needs to be enforced in the DRJCC.
%
\end{remark}

\begin{proposition}
For the DRJCC with the affine policy as the form (14), sufficient approximation conditions can be obtained as
\begin{align*}
&\begin{cases}
	\underset{\mathbb{P} _{\xi}\in \mathcal{D}}{\mathrm{inf}}\,\,\mathbb{P} _{\boldsymbol{\xi }}\left\{ \left( a_{i}^{+}\left( \boldsymbol{x} \right) +\delta \right) \boldsymbol{e}^T\boldsymbol{\xi }\le b_i\left( \boldsymbol{x} \right) \right\} \ge 1-\frac{\epsilon}{2}, \forall i\in \left[ N \right]\;\,\mathrm{(25a)}\\
	\underset{\mathbb{P} _{\xi}\in \mathcal{D}}{\mathrm{inf}}\,\,\mathbb{P} _{\boldsymbol{\xi }}\left\{ \left( a_{i}^{-}\left( \boldsymbol{x} \right) -\delta \right) \boldsymbol{e}^T\boldsymbol{\xi }\le b_i\left( \boldsymbol{x} \right) \right\} \ge 1-\frac{\epsilon}{2}, \forall i\in \left[ N \right]\;\,\mathrm{(25b)}\\
	(17a)-(17d).\\
\end{cases}
\end{align*}
\end{proposition}

\begin{proof}
By integrating transformations in (18)-(20) and \textit{Proposition} 1, the conclusion in \textit{Proposition} 2 can be easily proofed.
\end{proof}

Utilizing \textit{Proposition} 2, the DRJCC is effectively transformed into several individual DRCCs. 
For clarity in subsequent analyses, denote the feasible set of original DRJCC, the average Bonferroni approximation, and the proposed method as
$Z:=\left\{ \boldsymbol{x}:\left( 14 \right) \right\},Z_{B}:=\left\{ \boldsymbol{x}:(15) \right\},	Z_{p}:=\left\{ \boldsymbol{x}:(25),(17) \right\}$, respectively. Also, represent the feasible set of outer approximation of the DRJCC (14) as   
\begin{align*}
	Z_O:=\left\{ \boldsymbol{x}:\inf _{\mathbf{P} \in \mathcal{P}} \mathbb{P}\left\{a_i(x) \boldsymbol{e}^T\xi \leq b_i(x)\right\} \geq 1-{\epsilon}, \forall i \in[N] \right\}.
\end{align*}

\begin{remark}
	By integrating the specificity of the affine policy, the proposed DRJCC transformation method achieves a more rational allocation of violation probabilities $\epsilon_i$, thus reducing the over-conservation and enhancing the operational economic efficiency. 
More rigorous, there is $Z_O\supseteq Z\supseteq Z_{p}\supseteq Z_{B}$, when $N>2$ that is easily achievable.
Especially, when the signs of all $a_i(\boldsymbol{x})$ can be determined in advance, there is $Z_O=Z=Z_p \supseteq Z_{B}$.

\end{remark}

\subsubsection{Reformulation of individual DRCCs}


Under the guidance of \textit{Proposition} 2, DRJCC as formulated in (14) are efficiently decomposed into several individual DRCCs (25a)-(25b).
According to the Theorem 2 in \cite{chao2022LiLun}, individual DRCCs under the $\alpha$-unimodal ambiguity set (10) like form (26) can be equivalently transformed into (27a)-(27b) as
\begin{align*}
&\underset{\mathbb{P} _{\xi}\in \mathcal{D}}{\mathrm{inf}}\,\,\mathbb{P} _{\boldsymbol{\xi }}\left\{ a_i\left( \boldsymbol{x} \right) \boldsymbol{e}^T \boldsymbol{\xi}\le b_i\left( \boldsymbol{x} \right) \right\} \ge 1-\epsilon_i
\tag{26}
\\
&\Leftrightarrow b_i(\boldsymbol{x})\ge r_i\left\| \mathbf{\Sigma }_{0}^{1/2}a_i(\boldsymbol{x}) \right\| _2,
\tag{27a}
\\
&\quad \;\, r_i=\left\{ \begin{array}{l}
	\sqrt{\frac{\gamma _2}{\epsilon_i}}\left( \frac{\eta}{2}+1 \right) ^{-1/\eta},\quad \,\,\mathrm{if} \frac{\gamma _2}{\gamma _{1}^{\prime}}\left( \frac{\eta +2}{\eta +1} \right) ^2<\frac{1}{\epsilon} \tag{27b}
	\\
	\frac{\sqrt{\eta _{2}^{2}+\eta (\eta +2)\left( 1/\epsilon_i -1 \right) \eta _{1}^{2}}+(\eta +1)\eta _2}{(\eta +2)\tau}, \mathrm{otherwise}\\
\end{array} \right. ,
\end{align*}
where $\gamma_1^{\prime}=\min \left\{\gamma_1, \frac{\eta(\eta+2)}{(\eta+1)^2} \gamma_2\right\}$,
$\eta_1=\sqrt{\frac{\eta+2}{\eta} \gamma_2-\left(\frac{\eta+1}{\eta}\right)^2 \gamma_1^{\prime}}$,
$\eta_2\!=\!\frac{\eta+1}{\eta} \sqrt{\gamma_1^{\prime}}$, 
and 

$\tau =\left\{ 1-\epsilon _i-\frac{\left( \sqrt{\epsilon _i\eta _{2}^{2}+\eta (\eta +2)\left( 1-\epsilon _i \right) \eta _{1}^{2}}-\sqrt{\epsilon _i\eta _{2}^{2}} \right) ^2}{(\eta +2)^2\eta _{1}^{2}} \right\} ^{-1/\eta}$ .

Constraints (27a)-(27b) are SOC, which can be efficiently implemented by commercial solvers.
\vspace{-5pt}
\subsection{Reformulation of Nonconvex HCNG Constraints}



\subsubsection{Reformulation of Bilinear Terms with HVF}

The HVF variable $w$ introduces quadratic and cubic terms in equations (4b), (4l), (4o), and (4p), resulting in the nonconvex high-order equality constraints.
Although commercial solvers like Gurobi can yield exact solutions for general quadratic equality constraints, their spatial branching algorithm is extremely time-consuming when dealing with multiple such constraints.
In E-HCNG networks, the HVF $w_t$ should remain consistent across expressions to ensure uniform operations of different components.
Therefore, the BEA is adopted for HVF $w_t$ \cite{rui2019Participation}. 
%
%
Specifically, for $w_t\in[0,w^U]$, it can be approximated as 
\begin{align*}
	w_t=\bigtriangleup w\sum_{k=1}^K{2^{k-1}w_{t,k}^{B}} ,\tag{28}
\end{align*}
where $\bigtriangleup w=w^U/2^K$, and $K$ is a positive integer that controls the precision. $w_{k}^{B}$ is the binary variable.


So that the bilinear terms in the form $xw_t$ can be approximated by the auxiliary variable $z_t$, with the following constraints:
\vspace{-5pt}
\begin{align*}
	z_t&=\bigtriangleup w\sum_{k=1}^K{2^{k-1}z_{t,k}}, \tag{29a}
	\\
	-M\left( 1-w_{t,k}^{B} \right) &\le x-z_{t,k}\le M\left( 1-w_{t,k}^{B} \right), \tag{29b}
	\\
	-Mw_{t,k}^{B}&\le z_{t,k}\le Mw_{t,k}^{B}, \tag{29c}
\end{align*}
where the continuous variable $z_{t,k} = xw_{t,k}^{B}$, and $M$ is a sufficient large number. Through constraints (29a)-(29c), the bilinear terms in equations (4l), (4o), and (4p) can be linearized.

\subsubsection{Reformulation of Nonconvex HCNG Flow}
In the HCNG flow equations (4b), (9b), and (9c), the flow coefficient $C_{mn,t}$ varies inversely with the HVF $w_t$, as referenced in (4m) and (4n).
This results in intractable nonconvex cubic equality constraints in the HNCG flow. 
Furthermore, the value of $sgn$ can be linearized using the big-M method without destroying the formula \cite{yangIESDiaoPinDROJiHuiYueShuDiaoDu2022}.

First, the transformations of equations (4b) and (9b) are presented.
Due to their similar forms, (4b) is used as an example.
Given the nonconvex high-order equality constraints of (4b), 
the convex-concave procedure is employed, equally transforming the (4b) as follows \cite{cheng2018Convexa}:
\vspace{-5pt}
%
%
%
%
%
\begin{align*}
	\frac{{F_{mn,t}}^2}{C_{mn,t}}+{\pi _{n,t}^{}}^2&\le {\pi _{m,t}^{}}^2,
	\tag{30a}
	\\
	{\pi _{m,t}^{}}^2 &\le \frac{{F_{mn,t}}^2}{C_{mn,t}}+{\pi _{n,t}^{}}^2.
	\tag{30b}
\end{align*}
For clarity, denote that
\vspace{-7pt}
\begin{align*}
	C_{mn,t}=C_{mn}^{0}\frac{M^{0}}{M_t}, \tag{31}
\end{align*}
where $M_t$ is linear function about $w_t$.
So, for constraint (30a), the first term on the left-hand side can be rewritten as
\begin{align*}
	\frac{{F_{mn,t}}^2}{C_{mn,t}}=\frac{{F_{mn,t}}^2}{C_{mn}^{0}}\cdot \frac{M_t}{M^{0}}=\frac{1}{C_{mn}^{0}}\left( F_{mn,t}\sqrt{\frac{M_t}{M^{0}}} \right) ^2. \tag{32}
\end{align*}
Considering that the maximum HVF is limited to 20\% in most practices \cite{chen2020GuanDaoXianQing1}, and taking into account the values of $M^{Gas}$ and $M^{Hy}$, the value of $\sqrt{\frac{M_t}{M^{0}}}$ is closed to 1.
So, by the first-order Taylor expansion, equation (33) can be approximated as
\begin{align*}
	\frac{{F_{mn,t}}^2}{C_{mn,t}}=\frac{1}{C_{mn}^{0}}\left( F_{mn,t}\sqrt{\frac{M_t}{M^0}} \right) ^2\approx \frac{1}{C_{mn}^{0}}\left[ F_{mn,t}\left( \frac{M_t}{2M^0}+\frac{1}{2} \right) \right] ^2. \tag{33}
\end{align*}
Since $\frac{M_t}{2M^0}$ is a linear function about $w_t$, the expression $F_{mn,t}\left( \frac{M_t}{2M^0}+\frac{1}{2} \right)$ can be linearized by the BEA with auxiliary variable $z^{F_1}_{mn,t}$, as introduced in (29a)-(29c).
Following this, constraint (31a) is transformed into the SOC constraint as 
\vspace{-2pt}
\begin{align*}
	\frac{{z_{mn,t}^{F}}^2}{C_{mn,t}^{0}}+{\pi _{n,t}^{}}^2\le {\pi _{m,t}^{}}^2 .\tag{34}
\end{align*}

For constraint (31b), the linear approximation for the right-hand side can be expressed as
\begin{align*}
	\frac{{F_{mn,t}}^{2}}{C_{mn,t}^{}}+{\pi _{n,t}}^{2}\approx \frac{2F_{mn,t}^{r}F_{mn,t}-\left( F_{mn,t}^{r} \right) ^2}{C_{mn,t}^{}}+2\pi _{n,t}^{r}\pi _{n,t}-\left( \pi _{n,t}^{r} \right) ^2
	\\
	=\frac{2F_{mn,t}^{r}F_{mn,t}}{C_{mn}^{0}}\cdot \frac{M_t}{M^0}-\frac{\left( F_{mn,t}^{r} \right) ^2}{C_{mn}^{0}}\cdot \frac{M_t}{M^0}+2\pi _{n,t}^{r}\pi _{n,t}-\left( \pi _{n,t}^{r} \right) ^2, \tag{35}
\end{align*}
where $(F_{mn,t}^r, \pi_{n,t}^r)$ is reference point of linear approximation.
$F_{mn,t}\frac{M_t}{M^0}$ can be linearized by the BEA with auxiliary variable $z^{F_2}_{mn,t}$. And $\frac{M_t}{M^0}$ can be approximated through (28) with the auxiliary variable $k^{M}_{t}$.
After that, (30b) can be approximated by
\begin{align*}
	{\pi _{m,t}^{}}^2\le \frac{2F_{mn,t}^{r}F_{mn,t}z_{mn,t}^{F_2}-\left( F_{mn,t}^{r} \right) ^2k_{t}^{M}}{C_{mn}^{0}}+2\pi _{n,t}^{r}\pi _{n,t}-\left( \pi _{n,t}^{r} \right) ^2. \tag{36}
\end{align*}

\vspace{-8pt}
Following this, the nonconvex cubic equality constraint (4b) is replaced by (34), (35), and(36), along with corresponding constraints for auxiliary variables $z^{F_1}_{mn,t}$, $z^{F_2}_{mn,t}$, and $k^{M}_{t}$ as detailed in (28)-(29c).
The similar approximation can be applied to constraint (9b). 

Afterward, constraint (9c) is the last constraint to be transformed. 
First, for the bilinear term $\alpha_{mn,t}F_{mn,t}$, $\alpha_{m,t}^\pi\pi_{m,t}$, and $\alpha_{n,t}^\pi\pi_{n,t}$, the McCormick envelope approximation is employed \cite{lirong2021GaiJinMcCormickSongChi}, 
introducing the auxiliary variables $z_{mn,t}^{\alpha F}$, $z_{m,t}^{\alpha \pi}$, and $z_{n,t}^{\alpha \pi}$.
In combination with (31), constraint (9c) can be rewritten as
\vspace{-3pt}
\begin{align*}
	z_{mn,t}^{\alpha F}\frac{M_t}{M^0}=C_{mn,t}^{0}\left( z_{m,t}^{\alpha \pi}-z_{n,t}^{\alpha \pi} \right), \tag{38}
\end{align*}

\vspace{-3pt}
\noindent where $z_{mn,t}^{\alpha F}\frac{M_t}{M^0}$ can also be linearized by (29a)-(29c),
with the corresponding auxiliary variable $z^{\alpha F_1}_{mn,t}$. 

Following the above transformation, the formulation of the E-HNCG scheduling is transformed into a mixed-integer second-order cone programming problem, which can be effectively solved by commercial optimization solvers.

\vspace{-5pt}
\section{Case Studies}
The proposed method is tested on two systems: one with a 5-bus power network and a 7-node gas network, the other with a 118-bus power network and a 20-node gas network. 
The time horizon of the scheduling $T = 24$. 
The maximum HVF is set to 20\% \cite{chen2020GuanDaoXianQing1}.
The parameters $\gamma_1$, $\gamma_2$, and $\eta$ in (4) are set at 0.1, 1.1, and 1, respectively. 
The risk level of DRJCC is set to $\epsilon$ = 0.05. And set $K=3$.
All simulations are coded in Python 3.11 and solved by Gurobi 11.0, on an Intel Core i9-13900 CPU with 64 GB RAM PC. The MIPGap for Gurobi is set to 0.005.

\vspace{-10pt}
\subsection{Small-Scale Test System}

The small-scale test system consisting of a 5-bus power system and a 7-node HCNG system is illustrated in Fig. 1. 
The system includes one non-GFU, two GFUs, two PV stations, two electricity loads, two gas sources, two compressors, three gas loads, and one P2H station with hydrogen storage.
The total system electrical load, gas load, and PV output curves are shown in Figure \ref{fig_load}.
The means and standard deviations of PV total output forecast errors are set at 0 and 4 MW, respectively.
Detailed component parameters are given in \cite{case_data}.

\begin{figure}[!t]
	\centering
	\includegraphics[width=3.2in]{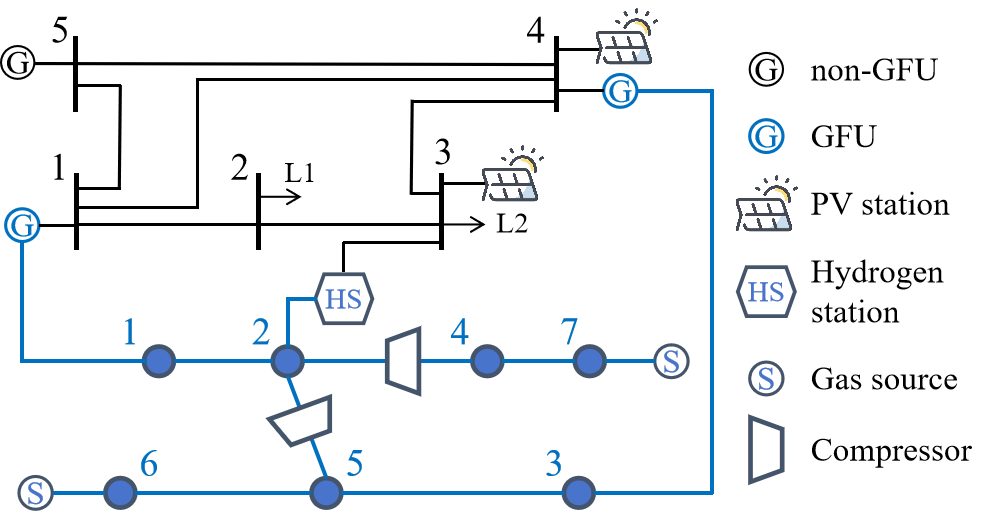}
	\caption{Topology of the small-scale E-HCNG network with a 5-bus power grid and a 7-node HCNG network.}
	\vspace{-12pt}
	\label{fig_load}
\end{figure}

\begin{figure}[!t]
	\centering
	\includegraphics[width=3.5in]{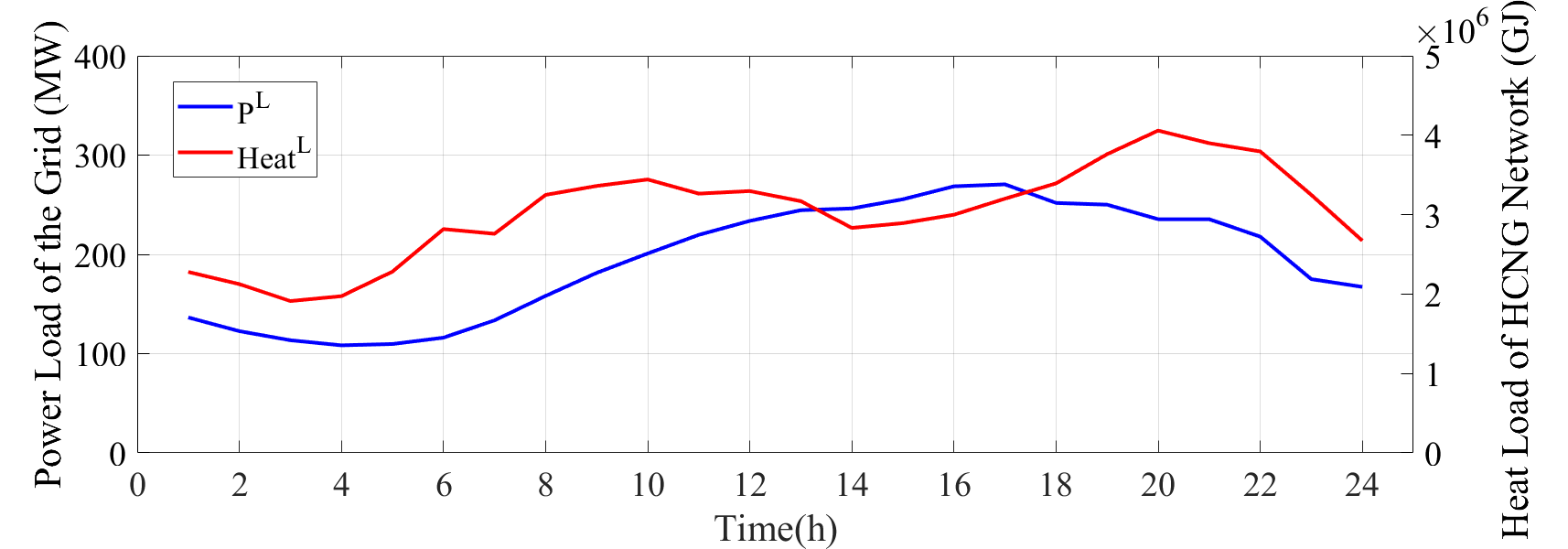}
	\caption{Profiles of electricity load and HCNG load converted to heat values.}
	\vspace{-12pt}
	\label{fig_load2}
\end{figure}

\subsubsection{Result of the proposed method}
The scheduling results of the E-HCNG network by the proposed method are shown in Fig. \ref{fig_test1}.
The total operating cost is $2.156*10^5\$$, with the optimal gap equal to 0. The computation time is 1238 seconds.



Fig. \ref{fig_test1}(a) illustrates the power balance of the power system.
The blue bars represent the PV output, which concentrates during daytime hours especially at noon.
The red bar is the non-GFUs output, which compensates for insufficient PV output.
For the P2H power, two points are worth mentioning: 
(1) During $t=9-14h$, excess power from the PV, exceeding the grid's load, is utilized by the P2H.
(2) During $t=7-8h$ and $t=15-18h$, the P2H still runs when the PV output exists but is less than the load.
This is due to P2H electrolyzers having higher flexibility compared to non-GFUs and GFUs. Thus, P2H runs at low power to provide reserves against PV output fluctuations.

Fig. \ref{fig_test1}(b) illustrates the gas volume and HVF in the HCNG network.
The red bars indicate the hydrogen volume released from storage, while the red line represents the HVF of the HCNG.
Notably, during $t=11-12h$, a significant hydrogen release corresponds with a gradual increase in HVF.
Conversely, the HVF gradually decreases during $t=17-18h$ when hydrogen release diminishes.
Because of the gas storage effect of line packs, the change of HVF is buffered with respect to the hydrogen release.
The gas storage capacity of line packs buffers the change of HVF relative to the hydrogen release.

Fig. \ref{fig_test1}(c) illustrates the operations of hydrogen storage. 
The blue bars represent the volume of hydrogen produced by P2H.
The red bars indicate the hydrogen released from storage.
The red line shows the volume of hydrogen in storage.
During $t=1-6h$, 
hydrogen storage is gradually depleted to maximize capacity for daytime P2H production.
Then, the storage rapidly increases during $t=9-10h$ as P2H production increases. 
The storage remains relatively stable during $t=11-18h$ and gradually diminishes during $t=19-24h$, optimizing the use of storage capacity.

\begin{figure}[!t]
	\centering
	\includegraphics[width=3.4in]{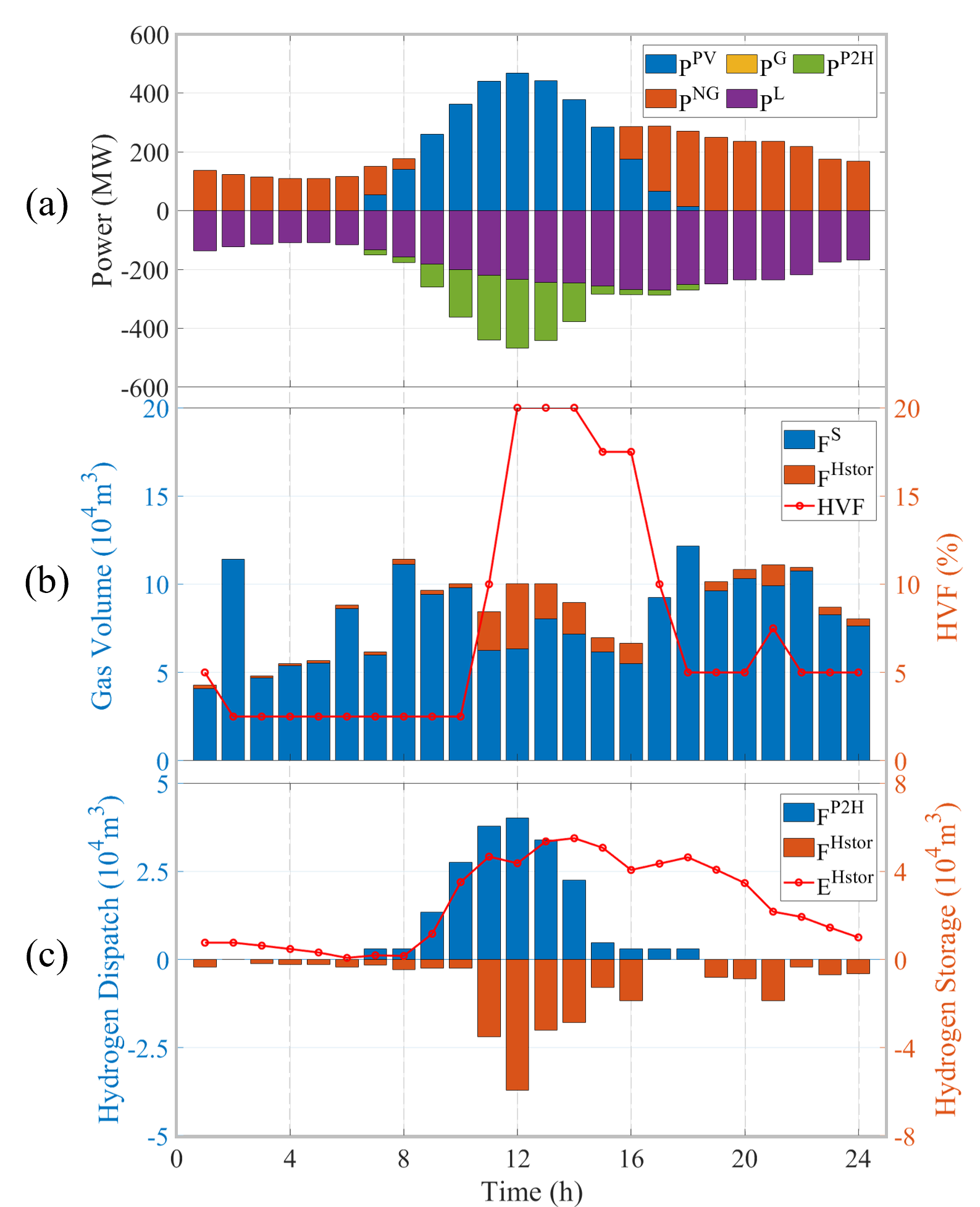}
	\caption{Operations of the small-scale E-HCNG network: (a) Operations and power balance of the grid. (b) HCNG supply and HVF value in the HCNG network. (c) Operations and stored volume of hydrogen storage. }
	\vspace{-12pt}
	\label{fig_test1}
\end{figure}
%

\subsubsection{Comparison with Fixed HVF Operations}
To illustrate the benefits of variable HVF operations, comparative experiments with fixed HVFs are conducted.
The fixed HVFs are set from 0\% to 20\% in steps of 2\%.
The operating costs at each fixed HVF are shown in Fig. \ref{fig_test2_HVF}.

\begin{figure}[!t]
	\centering
	\includegraphics[width=3.5in]{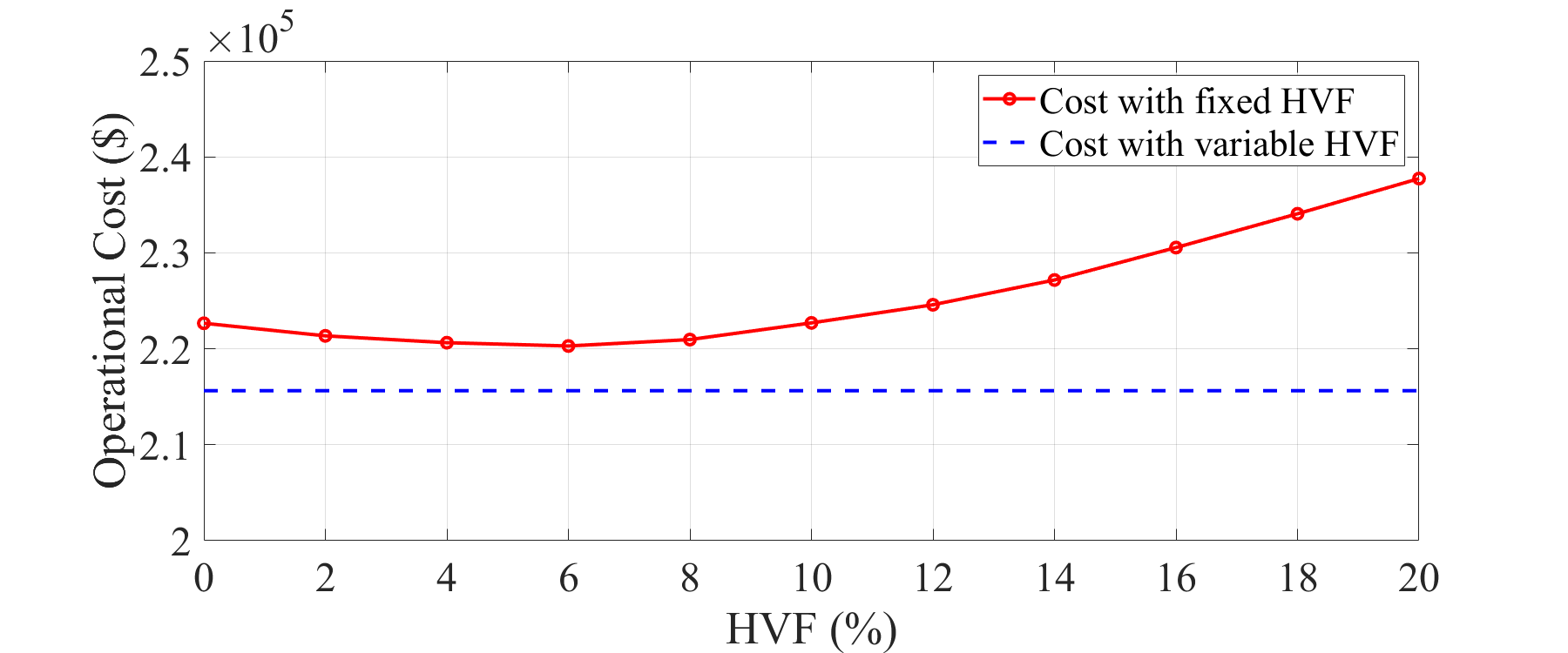}
	\caption{Operating costs under different HVF models.}
	\vspace{-15pt}
	\label{fig_test2_HVF}
\end{figure}
Among fixed HVF settings, the operating cost decreases and then increases as the HVF rises.
Under low HVFs, P2H and hydrogen blending effectively utilize the surplus PV power during peak hours.
However, when HVF exceeds $8\%$, the PV output fails to meet the P2H demands, necessitating the use of non-GFUs to supplement the power. This results in significantly higher costs.
The result demonstrates that HVF is crucial to operating costs.

Furthermore, by comparing variable and fixed HVF settings, the operating cost of the variable HVF is significantly lower.
Note that the 
This reflects the flexibility benefits of the variable HVF operation paradigm.

\begin{table}[htbp]
	\centering
	\caption{Operating costs and EJVPs of different affine policies}
	\label{tab_test3_compare}
	\begin{tabular}{c c c c c c}
		\toprule
		& \textit{P1} & \textit{P2} & \textit{P3} & \textit{P4} & \textit{P\textsuperscript{P}}\\
		\midrule
		Operating Cost ($10^5$\$) & 2.07 & 2.29 & infeasible & 2.17 & 2.16 \\ 
		EJVPs(\%) & 49.33 & 0.00 & N/A & 0.01 & 0.01\\ 
		\bottomrule
	\end{tabular}
\end{table}
\subsubsection{Comparison with Other Dispatch Policy}
To demonstrate the benefits of the proposed E-HCNG affine policy,
comparative experiments are conducted based on the following scheduling modes:

\textit{P1}: Deterministic scheduling policy, assuming PV outputs are deterministic and known in advance.

\textit{P2}: An affine policy with non-GFUs only.

\textit{P3}: An affine policy with P2H only.

\textit{P4}: An affine policy with P2H only, expanding hydrogen storage capacity to 1.5 times its original size.

\textit{P\textsuperscript{P}}: The proposed system-level E-HCNG affine policy.



Furthermore, to quantify the reliability of the scheduling strategy, the empirical joint violation probability (EJVP) is introduced. 
The EJVP is defined as the ratio of samples in which any constraint is violated \cite{lun2022YangLun}.
To generate samples about PV output predict error $\xi_t$, a Monte Carlo simulation is conducted using 10,000 samples via Gaussian distribution, with the maximum statistic error based on $\gamma_1$ and $\gamma_2$.

The operating cost and reliability results for different scheduling modes are shown in Table \ref{tab_test3_compare}.
The deterministic policy (\textit{P1}) exhibits the lowest operating cost, since it does not account for the reserve capacity costs. 
Consequently, this approach also results in the lowest reliability.

The affine policies \textit{P2}-\textit{P\textsuperscript{P}} guarantee reliability, with the proposed policy Policy \textit{P\textsuperscript{P}} having the lowest operating cost among them.
Policy \textit{P2} results in higher operating costs due to the higher reserve cost of non-GFUs.
Policy \textit{P3} is infeasible, since the affine policy relies solely on P2H and without storage release may exceed the storage capacity during adjustment.
To further evaluate the affine policy, we expand the hydrogen storage capacity, thus creating Policy \textit{P4}.
Although feasible, Policy \textit{P4} still has a higher operating cost compared to Policy \textit{P\textsuperscript{P}}.
Considering the substantial construction costs of hydrogen storage, these results show the superiority and necessity of the proposed E-HCNG affine policy.

\subsubsection{Comparison with other DRCC Methods}
As discussed in Section III-A, the proposed DRJCC method maintains the operational violation risk below the setting level, while avoiding the over-conservation caused by uncertainty characterization and the Bonferroni approximation.
To illustrate this, experiments are carried out based on the following methods:

\textit{M1}: DRJCC transformed by averaged Bonferroni approximation with $\epsilon_i=\frac{\epsilon}{N}$.

\textit{M2}: DRJCC transfromed by outer approximation, i.e., $\epsilon_i=\epsilon$.

\textit{M3}: DRJCC transformed by the proposed method but approximately solved using the normal quantile.

\textit{M4}: DRJCC transformed by the proposed method but built on a general second-order moment-based ambiguity set.

\textit{M\textsuperscript{P}}: DRJCC transformed and solved by the proposed method.

Similar to the previous part, 100,000 samples generated from the Gaussian distribution are used to check the reliability.
These samples are divided into ten groups, each containing 10,000 samples.
The maximum, minimum, and average EVJPs among groups are recorded.
The operating cost and reliability results are shown in Table \ref{tab_test4}.

The comparison of \textit{M\textsuperscript{P}} with \textit{M1} and \textit{M2} validates the superiority of the proposed method of converting DRJCC into individual DRCCs.
The averaged Bonferroni approximation method (\textit{M1}) has a significantly higher operating cost due to the over-conservation.
Although the outer approximation method (\textit{M2}) has a similar operating cost to the proposed method, this relaxation method cannot theoretically guarantee the operational risk level \cite{weijun2022LiLun}.

The comparison of \textit{M\textsuperscript{P}} with \textit{M3} and \textit{M4} validates the superiority of the unimodality-skewness informed DRCC.
In the stochastic programming method via quantile (\textit{M3}), the operational decisions fail to meet the setting risk level.
Additionally, compared with the general second-order moment-based DRCC (\textit{M4}), the proposed method results in a lower operating cost. This advantage is due to the additional information about uncertainties, which allows for more precise risk control.
These conclusions are further verified in the sensitivity experiments in the next part.

\begin{table}[htbp]
	\centering
	\caption{Operating costs and EJVPs of different chance-constrained methods}
	\label{tab_test4}
	\begin{tabular}{c c c c c c c}
		\toprule
		& \textit{M1} & \textit{M2} & \textit{M3} & \textit{M4} & \textit{M\textsuperscript{P}} \\
		\midrule
		Operating Cost ($10^5$\$)  & 2.50 & 2.14 & 2.10 & 2.19 & 2.16\\ 
		Maximum EVJPs(\%)  & 0.00 & 0.41 & 7.37 & 0.00 & 0.01\\ 
		Minimun EVJPs(\%) & 0.00 & 0.45 & 7.26 & 0.00 & 0.00 \\ 
		Average EVJPs(\%) & 0.00 & 0.43 & 7.33 & 0.00 & 0.01 \\ 
		\bottomrule
	\end{tabular}
\end{table}

\subsubsection{Sensitivity of DRJCC methods to risk levels}
To further validate the effectiveness of the proposed method under different risk levels, sensitivity experiments are conducted.
The operating costs and the average EJVPs are shown in Fig. \ref{fig_test5} and Table \ref{tab_test5}, respectively.

In Table \ref{tab_test5}, it is evident that the EJVPs for \textit{M\textsuperscript{P}} are below the preset risk levels under all conditions, validating the effectiveness of the proposed method.
For \textit{M1}, the EJVPs are consistently 0.00\% across all risk levels, which illustrates the over-conservation of Bonferroni approximation for E-HCNG networks.
The over-conservation results in significantly higher operating costs, as shown in Fig. \ref{fig_test5}.
Furthermore, \textit{M2} and \textit{M3} exhibit substantially higher EJVPs.
Note that both these two methods cannot guarantee the operational risk level, although the EJVP for \textit{M2} is below the setting risk. 
Lastly, the EJVPs for \textit{M4} are also below the setting risk levels, but Fig. \ref{fig_test5} reveals differences in operating costs.

In Fig. \ref{fig_test5}, it is evident that larger risk levels $\epsilon$ lead to lower operating costs across all methods, particularly noticeable when $\epsilon$ is small.
Specifically, the proposed method (\textit{M\textsuperscript{P}}) exhibits the lowest operating cost among methods capable of guaranteeing operational risk levels (\textit{M\textsuperscript{P}}, \textit{M1}, and \textit{M4}, as previously analyzed).
Moreover, the cost difference between the proposed method (\textit{M\textsuperscript{P}}) and the outer approximation (\textit{M2}) is minimal.
This minimal cost variation is attributed to the decreased effectiveness of the DRJCC constraints as $\epsilon_i$ increases.
Most importantly, the proposed method can safeguard the operational risk, while method \textit{M2} and \textit{M3} cannot, thereby emphasizing the advantages of the proposed method.
%
%
\begin{table}[!t]
	\centering
	\caption{EJVPs under various risk level $\epsilon$}
	\label{tab_test5}
	\begin{tabular}{c c c c c c c}
		\toprule
		$\epsilon$& 0.05 & 0.10 & 0.15 & 0.20 & 0.25 & 0.30 \\
		\midrule
		\textit{M1} & 0.00 & 0.00 & 0.00 & 0.00 & 0.00 & 0.00 \\ 
		\textit{M2} & 0.43 & 4.34 & 9.91 & 15.34 & 20.21 & 24.43\\ 
		\textit{M3} & 7.33 & 13.31 & 25.66 & 24.27 & 29.44 & 34.55\\
		\textit{M4} & 0.00 & 0.00 & 0.04 & 0.27 & 0.67 & 1.37 \\ 
		\textit{M\textsuperscript{P}} & 0.01 & 0.42 & 1.96 & 4.35 & 7.10 & 9.99 \\ 
		\bottomrule
	\end{tabular}
\end{table}
\begin{figure}[!t]
	\centering
	\includegraphics[width=3.4in]{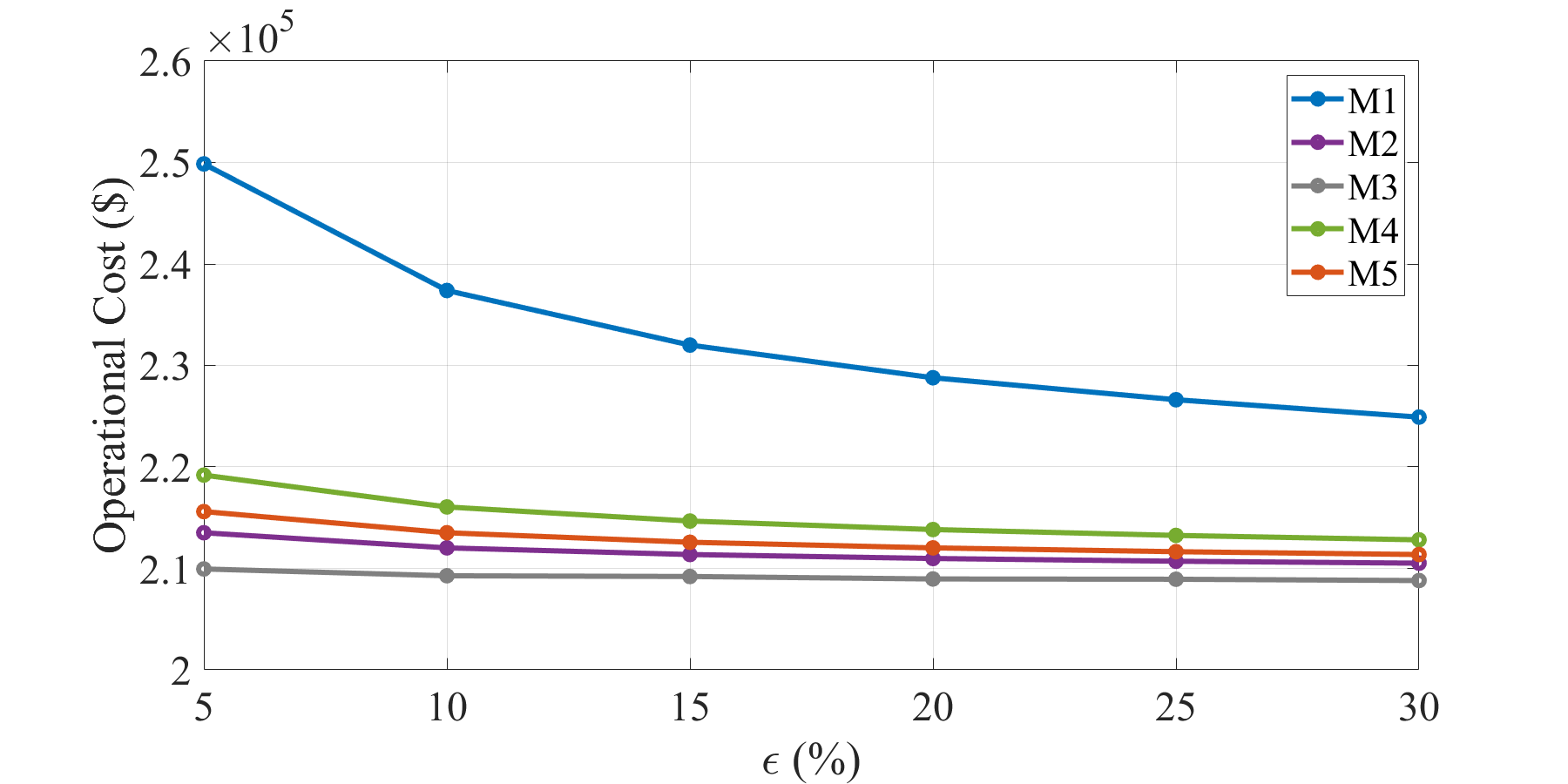}
	\caption{Operating cost under various risk level $\epsilon$.}
	\vspace{-13pt}
	\label{fig_test5}
\end{figure}
\vspace{-10pt}
\subsection{Big-Scale Test System}
To demonstrate the performance and efficiency of the proposed method on larger-scale systems, experiments are conducted on a system comprising a 118-bus power grid and a 20-node HCNG network.
The system includes 42 non-GFUs, 4 GFUs, 12 PV stations, 91 electricity loads, two gas sources, two compressors, nine gas loads, and 4 P2H stations with hydrogen storage.
The topology and parameters of the large-scale system are detailed in \cite{case_data}.


The operating cost of the entire system is $1.57*10^6$\$, with a computation time of 13691 seconds.
The detailed operations are illustrated in the Fig. 7,
where the scheduling strategy is similar to the small-scale system.

\begin{figure}[!t]
	\centering
	\includegraphics[width=3.4in]{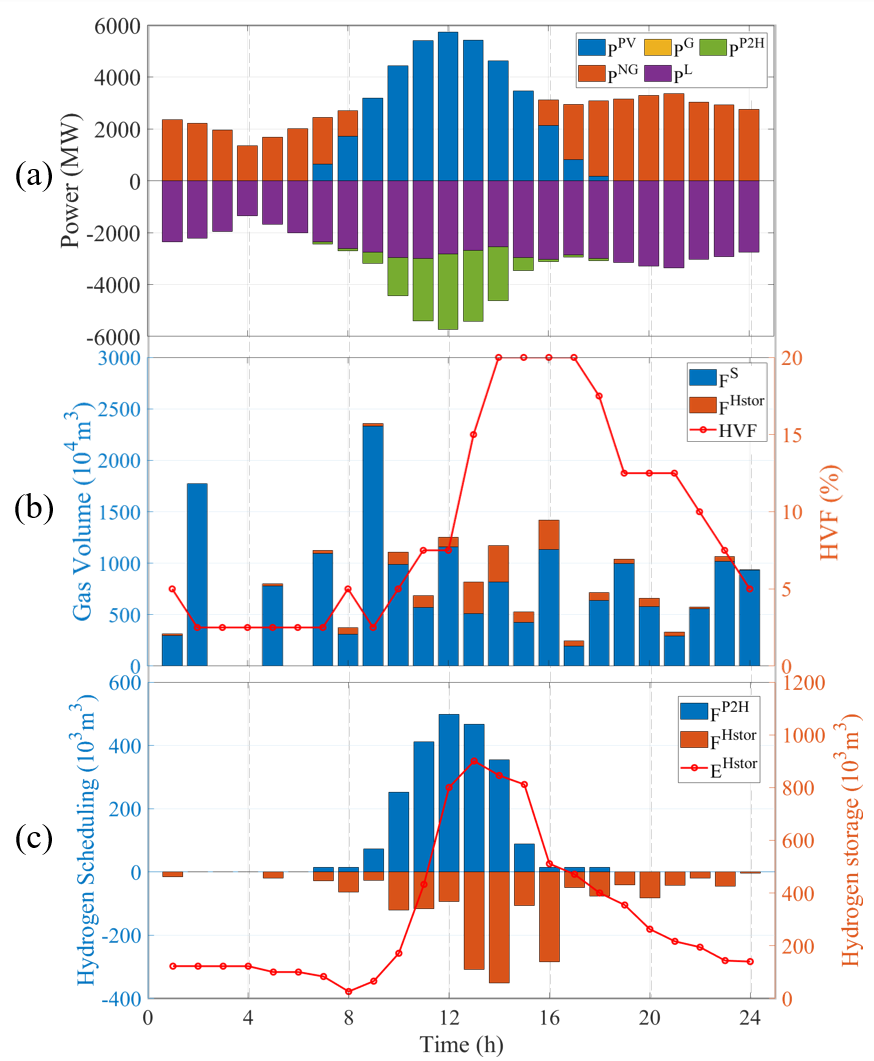}
	\caption{Operations of the large-scale E-HCNG network: (a) Operations and power balance of the grid. (b) HCNG supply and HVF value in the HCNG network. (c) Operations and stored volume of hydrogen storage. }
	\vspace{-12pt}
	\label{fig_test_big}
\end{figure}

\begin{table*}[htbp]
	\centering
	\caption{Operating costs and EJVPs of different operation approach in the large-scale system}
	\label{tab_test_big}
	\begin{tabular}{c c c c c c c c c c}
		\toprule
		& \textit{P0} & \textit{P1} & \textit{P2} & \textit{P3} & \textit{M1} & \textit{M2} & \textit{M3} & \textit{M4}& \textit{P\textsuperscript{P}}/\textit{M\textsuperscript{P}}\\
		\midrule
		Operating Cost ($10^6$\$)  & 1.69 & 1.52 & 1.64 & 1.65 & infeasible & 1.56 & 1.54 & 1.59  & 1.57\\ 
		EJVPs(\%) & 0.01 &  49.33 & 0.01 & 0.00 &  N/A  & 0.27 & 10.85 &  0.00 & 0.01 \\ 
		\bottomrule
	\end{tabular}
\end{table*}

Comparative experiments are also conducted on large-scale systems.
Additionally, let \textit{P0} denote the operational strategy based on fixed HVF with an optimal value, which can be obtained by experiments as described in Section IV-A.2).
The results are shown in Table \ref{tab_test_big}.

Comparison of \textit{M1} and \textit{P0} shows the significant advantage of variable HVFs in terms of operating costs.
Although the deterministic strategy (\textit{P1}) has a lower cost due to neglecting reserves, the high EJVP renders its reliability unacceptable.
For the appointed affine strategies \textit{P2} and \textit{P3}, they are less flexible compared to the system-level affine policy (\textit{P1}), resulting in higher operating costs.

The comparison between \textit{M\textsuperscript{P}} and \textit{M1}, \textit{M2} reconfirms the superiority of the proposed DRJCC processing method. 
The average Bonferroni approximation for DRJCC (\textit{M1}) is over-conservative under numerous conditions, making the problem infeasible as the system cannot provide the reserve capacities that \textit{M1} requires.
Additionally, the outer approximation for DRJCC (\textit{M2}) cannot provide theoretical guarantees for the operational risk level.
Further comparisons between \textit{M\textsuperscript{P}} and \textit{M3}, \textit{M4} reconfirm the superiority of the proposed DRCC model.
Under the same joint chance constraint processing method, the distribution quartile-based method (\textit{M3}) fails to meet the operational risk level, underscoring the necessity of DRCC for reliable operations.
Moreover, the general second-order moment-based DRCC (\textit{M4}) has higher operating costs while guaranteeing the same risk level. This reflects the benefits of an unimodality-skewness informed ambiguity set, which avoids unnecessary over-conservation.

These experiments collectively demonstrate the superiority of the proposed methodology, which achieves more economical performance while safeguarding operational risk.



\section{Conculsion}
Faced with the challenge of renewable energy fluctuation in E-HCNG network operations, this paper proposes a flexible and reliable scheduling method.
Flexibility is achieved through the HVF and a system-level re-dispatch affine policy.
Reliability is ensured by restricting the re-dispatch volume under the renewable energy uncertainty by a DRJCC.
In the solving stage, first, to reduce the over-conservation in the DRJCC transformation, an improved risk allocation method is proposed, utilizing the correlations of violations under the affine policy. 
The proposed method can improve the operating economy, especially when DRJCC involves numerous conditions like E-HCNG networks, while ensuring the violation risk is within the setting level.
%
Moreover, to tackle the non-convexities introduced by variable HVF, a series of transform methods are proposed, and ultimately transform the problem into mixed-integer second-order cone programming.

Both small- and large-scale experiments validate the effectiveness of the proposed method.
The variable HVF and system-level affine policy effectively reduce operating costs and enhance reliability.
Furthermore, among several joint chance-constraint methods, the proposed method achieves the lowest operating cost while guaranteeing the set risk level.
Notably, the proposed risk allocation method is also applicable to other DRJCC problems with affine policies, potentially contributing further to other fields.
\vspace{-8pt}

\bibliographystyle{IEEEtran}
\bibliography{Paper2}

%

\vfill

\end{document}